\documentclass{cta-author}
\usepackage{textcomp}
\usepackage{appendix}

\usepackage{graphicx}
\usepackage{xspace}
\usepackage{float}
\usepackage{booktabs}
\usepackage{url}
\usepackage[
    hyperfootnotes=false,
    colorlinks=true,
    linkcolor=blue,
    urlcolor=blue,
    citecolor=blue,
    anchorcolor=blue,
    pagebackref=false]{hyperref}

\usepackage[dvipsnames]{xcolor}

\usepackage{amsmath,amssymb,amsfonts}
\usepackage{mathrsfs}
\usepackage{url}
\usepackage{empheq}
\usepackage{caption}
\usepackage{subcaption}
\usepackage[capitalize]{cleveref}

\newtheorem{lemma}{Lemma}
\newtheorem{definition}{Definition}
\newtheorem{theorem}{Theorem}
\newtheorem{remark}{Remark}
\newtheorem{proposition}{Proposition}
\newtheorem{corollary}{Corollary}
\newtheorem{assumption}{Assumption}
\newtheorem{condition}{Condition}

\usepackage{enumitem}

\usepackage{tikz}
\usepackage{xspace}

\DeclareRobustCommand{\legendline}[1]{\hspace{-2pt}
\tikz[#1,line width=1.5pt,baseline=-0.5ex]{\draw (0,0) -- (.35,0);}
\hspace{-2pt}}

\definecolor{mblue}{rgb}{0,0.4470,0.7410}
\definecolor{morange}{rgb}{0.8500,0.3250,0.0980}
\definecolor{myellow}{rgb}{0.9290,0.6940,0.1250}
\definecolor{mpurple}{rgb}{0.4940,0.1840,0.5560}
\definecolor{mgreen}{rgb}{0.4660,0.6740,0.1880}
\definecolor{mcyan}{rgb}{0.3010,0.7450,0.9330}
\definecolor{mred}{rgb}{0.6350,0.0780,0.1840}
\definecolor{mgreenblue}{rgb}{0.0,1.0,0.5}
\definecolor{parulablue}{rgb}{0.2431,0.1490,0.6588}
\definecolor{parulalblue}{RGB}{39,151,235}
\definecolor{parulagreen}{RGB}{129,204,89}
\definecolor{parulayellow}{RGB}{249,251,21}
\definecolor{wintergreen}{cmyk}{0.61,0,0.74,0}

\DeclareFontFamily{OT1}{pzc}{}
\DeclareFontShape{OT1}{pzc}{m}{it}{ <-> s*[1.1] pzcmi7t }{}
\DeclareMathAlphabet{\mathpzc}{OT1}{pzc}{m}{it}

\newcommand{\mc}[1]{\mathcal{#1}}

\newcommand{\mr}[1]{\mathrm{#1}}
\newcommand{\mf}[1]{\mathfrak{#1}}
\newcommand{\mb}[1]{\mathbb{#1}}
\newcommand{\ms}[1]{\mathscr{#1}}

\newcommand{\mpz}[1]{\mathpzc{#1}}

\newcommand{\Partial}[2]{\frac{\partial #1}{\partial #2}}

\newcommand{\C}[1]{\ensuremath{\mc{C}_{#1}}\xspace}   %

\newcommand{\dlp}[1]{\ensuremath{\ell_{#1}}\xspace}

\newcommand{\dltwo}{\dlp{2}}

\newcommand{\dlsp}[1]{\ensuremath{\ell_{\mr{s}#1}}\xspace}

\newcommand{\dlstwo}{\dlsp{2}}

\newcommand{\posClassSym}{\mc{Q}}
\newcommand{\posClass}[1]{\posClassSym_{#1}}
\newcommand{\posClassI}{\posClassSym_\mr{i}}

\newcommand{\reals}{{\mb{R}}}
\newcommand{\preals}{{\reals^+}}
\newcommand{\nnreals}{{\reals_0^+}}

\newcommand{\integ}{\ensuremath{\mb{Z}}\xspace}
\newcommand{\nninteg}{{\ensuremath{\integ_0^+}}}

\newcommand{\sym}{\ensuremath{\mb{S}}\xspace}

\DeclareMathOperator{\col}{col}

\DeclareMathOperator*{\arginf}{arg\,inf}
\DeclareMathOperator{\proj}{\pi}

\newcommand{\norm}[1]{\left\lVert#1\right\rVert}

\newcommand{\st}{x}		          %
\newcommand{\stSet}{\mc{X}}      %
\newcommand{\stSize}{{n_\mr{x}}}   %

\newcommand{\stIc}{\st_0}

\newcommand{\op}{y}		%

\newcommand{\opSize}{{n_\mr{y}}}

\newcommand{\gd}{w}   %
\newcommand{\gdSet}{\mc{W}}
\newcommand{\gdSize}{{n_\mr{w}}}

\newcommand{\gp}{z}   %

\newcommand{\gpSize}{{n_\mr{z}}}

\newcommand{\stEq}{\st_*}

\newcommand{\opEq}{\op_*}
\newcommand{\gdEq}{\gd_*}
\newcommand{\gpEq}{\gp_*}

\newcommand{\eqSet}{\ms{E}}
\newcommand{\stSetEq}{\ms{X}}
\newcommand{\gdSetEq}{\ms{W}}
\newcommand{\gpSetEq}{\ms{Z}}

\newcommand{\eqMap}{\kappa}

\newcommand{\perf}{\gamma}  %

\newcommand{\slackA}{\alpha}

\newcommand{\stMap}{f} 
\newcommand{\opMap}{h}

\newcommand{\B}{\mf{B}}  %
\newcommand{\Bw}{\mf{B}_\mr{\gd}} %

\newcommand{\stip}{\xi}

\newcommand{\var}{\lambda}

\newcommand{\difA}{A_\delta}
\newcommand{\difB}{B_\delta}
\newcommand{\difC}{C_\delta}
\newcommand{\difD}{D_\delta}

\newcommand{\sttran}{\phi_\mr{\st}}

\newcommand{\stdot}{\st_\mr{v}}
\newcommand{\stbdot}{\stb_\mr{v}}

\newcommand{\stdotSet}{\mc{D}}

\newcommand{\ltiB}{B}
\newcommand{\ltiC}{C}

\newcommand{\lpvA}{A}
\newcommand{\lpvB}{B}
\newcommand{\lpvC}{C}
\newcommand{\lpvD}{D}

\newcommand{\sch}{p}		%
\newcommand{\schSet}{\mc{P}}
\newcommand{\schSize}{{n_\mr{p}}}

\newcommand{\schMap}{\eta}

\newcommand{\stSetLPV}{\mpz{X}}
\newcommand{\gdSetLPV}{\mpz{W}}

\newcommand{\storquad}{M}  %

\newcommand{\schdot}{v}
\newcommand{\schdotSet}{\Pi}

\newcommand{\lyapfun}{V}
\newcommand{\supfun}{s}
\newcommand{\storfun}{\mc{V}}

\newcommand{\supQ}{Q}
\newcommand{\supR}{R}
\newcommand{\supS}{S}

\newcommand{\qsr}{(\supQ,\supS,\supR)}

\newcommand{\qsrMat}{\begin{bmatrix} \supQ  & \supS\\ \star & \supR \end{bmatrix}}

\newcommand{\dst}{\st_\delta}

\newcommand{\dgd}{\gd_\delta}
\newcommand{\dgp}{\gp_\delta}

\newcommand{\lyapfunDif}{\lyapfun_\delta}
\newcommand{\supfunDif}{\supfun_\delta}
\newcommand{\storfunDif}{\storfun_\delta}

\newcommand{\lyapfunIncr}{\lyapfun_\mr{i}}
\newcommand{\supfunIncr}{\supfun_\mr{i}}
\newcommand{\storfunIncr}{\storfun_\mr{i}}

\newcommand{\otherTraj}[1]{\expandafter\tilde #1}
\newcommand{\sto}{\otherTraj{\st}}

\newcommand{\gdo}{\otherTraj{\gd}}
\newcommand{\gpo}{\otherTraj{\gp}}

\newcommand{\stoIc}{\sto_0}

\newcommand{\parTraj}[1]{\expandafter\bar #1}
\newcommand{\stb}{\parTraj{\st}}

\newcommand{\gdb}{\parTraj{\gd}}
\newcommand{\gpb}{\parTraj{\gp}}

\newcommand{\stbIc}{\stb_0}

\newcommand{\pathSet}{\Gamma}
\newcommand{\stbmin}{\chi}

\newcommand{\Bv}{\B_\mr{v}}
\newcommand{\Bvw}{\B_\mr{v,\gd}}
\newcommand{\Bvset}[1]{\B_\mr{v,#1}}

\newcommand{\lyapfunShift}{\lyapfun_\mr{s}}
\newcommand{\supfunShift}{\supfun_\mr{s}}
\newcommand{\storfunShift}{\storfun_\mr{s}}

\newcommand{\lyapfunVelo}{\lyapfun_\mr{v}}
\newcommand{\supfunVelo}{\supfun_\mr{v}}
\newcommand{\storfunVelo}{\storfun_\mr{v}}

\newcommand{\velA}{A_\mr{v}}
\newcommand{\velB}{B_\mr{v}}
\newcommand{\velC}{C_\mr{v}}
\newcommand{\velD}{D_\mr{v}}

\newcommand{\gpdot}{\gp_\mr{v}}

\newcommand{\exoA}{A_\mr{\gd}}
\newcommand{\exoBvr}{\mf{W}}

\newcommand{\ddelta}[1]{#1_{\Delta}}
\newcommand{\dnabla}[1]{#1_{\nabla}}

\newcommand{\dtst}{\ddelta{\st}}
\newcommand{\dtgd}{\ddelta{\gd}}
\newcommand{\dtgp}{\ddelta{\gp}}
\newcommand{\bdtst}{\dnabla{\st}}

\newcommand{\bdtgp}{\dnabla{\gp}}

\newcommand{\stpl}{\st_+}
\newcommand{\gdpl}{\gd_+}

\newcommand{\deltB}{\Delta}

\newcommand{\vintA}{\bar{A}_\mr{v}}
\newcommand{\vintB}{\bar{B}_\mr{v}}
\newcommand{\vintC}{\bar{C}_\mr{v}}
\newcommand{\vintD}{\bar{D}_\mr{v}}

\newcommand{\customurl}[2]{\href{#1}{\path{#2}}}
\newcommand{\proofsection}[1]{\subsection{Proof of \cref{#1}}}

\allowdisplaybreaks

\makeatletter
\g@addto@macro\normalsize{%
  \setlength\abovedisplayskip{.5em}
  \setlength\belowdisplayskip{.5em}
  \setlength\abovedisplayshortskip{.5em}%
  \setlength\belowdisplayshortskip{.5em}%
}

\let\olddot\dot 
\gdef\dot{\expandafter\olddot}

\let\oldddot\ddot 
\gdef\ddot{\expandafter\oldddot}

\crefformat{equation}{(#2#1#3)}   %
\crefrangeformat{equation}{(#3#1#4)--(#5#2#6)}
\crefmultiformat{equation}{(#2#1#3)}{ and~(#2#1#3)}{, (#2#1#3)}{, and~(#2#1#3)}
\crefrangemultiformat{equation}{(#3#1#4)--(#5#2#6)}{ and~(#3#1#4)--(#5#2#6)}{, #3#1#4--(#5#2#6)}{, and~#3#1#4--(#5#2#6)}

\crefformat{assumption}{Assumption~#2#1#3}
\crefrangeformat{assumption}{Assumptions~#3#1#4--#5#2#6}
\crefmultiformat{assumption}{Assumptions~#2#1#3}{ and~#2#1#3}{, #2#1#3}{, and~#2#1#3}
\crefrangemultiformat{assumption}{Assumptions~#3#1#4--#5#2#6}{ and~#3#1#4--#5#2#6}{, #3#1#4--#5#2#6}{, and~#3#1#4--#5#2#6}

\begin{document}
\supertitle{Manuscript submitted to IET Control Theory \& Applications}

\title{Convex Equilibrium-Free Stability and Performance Analysis of Discrete-Time Nonlinear Systems$^\dagger$}

\author{\au{P.J.W. Koelewijn\,$^{1}$}, \au{S. Weiland\,$^1$}, \au{R. T\'{o}th\,$^{1,2,*}$}}

\address{%
\add{1}{Control Systems Group, Department of Electrical Engineering, Eindhoven University of Technology, The Netherlands}
\add{2}{Systems and Control Laboratory, HUN-REN Institute for Computer Science and Control, Hungary}
\email{r.toth@tue.nl}}

\begin{abstract}
	This paper considers the equilibrium-free stability and performance analysis of discrete-time nonlinear systems. We consider two types of equilibrium-free notions. Namely, the universal shifted concept, which considers stability and performance w.r.t. all equilibrium points of the system, and the incremental concept, which considers stability and performance between trajectories of the system. In this paper, we show how universal shifted stability and performance of discrete-time systems can be analyzed by making use of the time-difference dynamics. Moreover, we extend the existing results for incremental dissipativity for discrete-time systems based on dissipativity analysis of the differential dynamics to more general state-dependent storage functions for less conservative results. Finally, we show how both these equilibrium-free notions can be cast as a convex analysis problem by making use of the linear parameter-varying framework, which is also demonstrated by means of an example. \vspace{-1mm}
\end{abstract}

\maketitle

\footnotetext[2]{\normalfont This work has received funding from the European Research Council (ERC) under the European Union Horizon 2020 research and innovation programme (grant agreement No 714663) and was also supported by the European Union within the framework of the National Laboratory for Autonomous Systems (RRF-2.3.1-21-2022-00002).}

\section{Introduction}\label{8_sec:introduction}

The analysis of nonlinear systems has been an important topic of research over the last decades as systems have become more complex due to the push for higher performance requirements. While a large part of the existing tools for analyzing stability and performance of nonlinear systems, such as Lyapunov stability \cite{Khalil2002} and dissipativity \cite{Willems1972,VanderSchaft2017}, are very powerful, they only analyze these properties w.r.t. a single (equilibrium) point in the state-space representation, often the origin of the state-space. This poses limitations in cases when analyzing the stability and performance of a system w.r.t. multiple equilibria or even w.r.t. multiple trajectories, such as in reference tracking and disturbance rejection, is required. Therefore, in recent years, equilibrium-free stability and performance methods have become increasingly popular. These equilibrium-free methods include  concepts such as universal shifted and incremental stability and performance. Universal shifted stability and performance, also referred to as equilibrium independent stability and performance, considers these notions w.r.t. all equilibrium points of the system \cite{Simpson-Porco2019, Hines2011, Koelewijn2023}. Another, stronger, equilibrium-free concept is that of incremental stability and performance, which considers stability and performance between trajectories of the systems \cite{Verhoek2020, Pavlov2008, Forni2013a, Angeli2002,Tran2016}. This means that in case of incremental stability, trajectories of the system converge towards each other. Often, incremental stability and performance is analyzed through so-called contraction analysis \cite{Forni2013a, Tran2018, Lohmiller1998}.

While \emph{Continuous-Time} (CT) dynamical systems are often used for analysis and control, the recent surge of developments on data and learning based analysis and control methods heavily relies on \emph{Discrete-Time} (DT) formulations.
Also, controllers are almost exclusively implemented with digital hardware, hence analysis of the implemented discretized form of the controller with the actuator and sampling dynamics requires DT analysis tools, which are also the first step towards reliable DT synthesis methods. Therefore, analysis of DT nonlinear systems is inevitably important.

 For CT nonlinear systems, it has been shown how the velocity form of the system, i.e., the time-differentiated dynamics, imply universal shifted stability and performance \cite{Koelewijn2023}. In DT, the time-difference dynamics, analogous to the time-differentiated dynamics in CT, have primarily received attention in the context (nonlinear) model predictive control methods. In these works, the time-difference dynamics have been mainly used to realize offset free tracking of reference signals \cite{Cisneros2016,Ferramosca2009}. 
However, to the authors' knowledge, there are no results in literature which connect the time-difference dynamics to stability and performance guarantees w.r.t. equilibrium points in the nonlinear context and how to cast the corresponding analysis problem as computable tests. Therefore, in this paper, we show how these concepts are connected and how the universal shifted stability and performance analysis problem of DT systems can be solved as a convex optimization problem.

Similarly, w.r.t. incremental stability and performance, most literature has focussed on CT systems. For CT systems, it has been shown how the so-called differential form -- representing the dynamics of the variations along trajectories -- can be used in order to imply incremental stability, dissipativity, and performance properties \cite{Verhoek2020, Forni2013a}. There have been some results on incremental and contraction based stability of DT systems, see e.g. \cite{Tran2016,Tran2018}, or focussing on Lipschitz properties \cite{Revay2023, Wang2023,Pauli2021,Fazlyab2019}. Moreover, for DT nonlinear systems, the link to incremental dissipativity based on the differential form has been made in \cite{Koelewijn2021a} for quadratic supply functions. However, the work in \cite{Koelewijn2021a} only considers quadratic storage functions depending on a constant matrix. This is more conservative compared to results available for CT systems, where the links between dissipativity of the differential form and incremental dissipativity has been made for more general storage functions. Therefore, in this paper, we will provide a novel generalization of the CT incremental stability and performance results for DT systems and show how these results can be cast as a convex optimization problem.

To summarize, the main contributions of this paper are as follows
\begin{enumerate}[label={\bfseries C\arabic*:\;},ref={C\arabic*}, left=\parindent] \vspace{-4mm}
\item Show how stability and performance properties of the time-difference dynamics of a DT nonlinear system imply universal shifted stability and performance properties of  the system. (\cref{8_thm:velotoshiftstab,8_thm:veloshiftperf})
\item Extend existing CT incremental dissipativity analysis results for DT nonlinear systems based on the differential form, enabling the use of state-dependent storage functions. (\cref{7_an_thm:ind-incr-dissip})
\item Show that both the universal shifted and incremental analysis problems can be cast as an analysis problem of a \emph{Linear Parameter-Varying} (LPV) representation. This allows these problems to be solved via convex optimization in terms of \emph{Semi-Definite Programs} (SDPs), providing computable tests for equilibrium-free stability and performance analysis of DT nonlinear systems. (\cref{7_an_thm:diffdissipLPV})
\end{enumerate} \vskip -3mm

The paper is structured as follows, in \cref{8_sec:shiftanalysis}, we give an overview of the definitions of universal shifted and incremental stability and performance. Next, in \cref{8_sec:veloanalysis}, we will introduce velocity based analysis for DT nonlinear systems and show how it can be used in order to imply universal shifted stability and performance properties. In \cref{sec:diffanalysis}, we will show how dissipativity analysis of the differential form in DT can be used in order to imply incremental dissipativity. Then, in \cref{sec:veloanddiff}, we discuss the connections between the velocity and differential based analysis results and the relation between universal shifted and incremental stability and performance. \cref{7_an_sec:lpv} shows how the velocity and differential analysis results can be cast as analysis problems of an LPV representation, enabling to solve the equilibrium-free analysis problem via convex optimization in terms of SDPs. In \cref{7_an_sec:example}, the usefulness of the developed analysis tools is demonstrated on an example. Finally, the conclusions are drawn in \cref{8_sec:concl}.\\

\noindent \textbf{Notation:} 
$\reals$ is the set of real numbers, while $\nnreals$ is the set of non-negative reals. $\nninteg$ is the set of non-negative integers. We denote by $\sym^n$ the set of real symmetric matrices of size $n\times n$. Denote the projection operation by $\proj$, s.t. for $\mc{D}=\mc{A}\times\mc{B}$, $a\in\proj_\mr{a}\mc{D}$ if and only if 
$\exists b \in \mc{B}$ s.t. $(a,b)\in\mc{D}$. For a signal $\gd:\nninteg\to\reals^n$ and a $\gd_*\in\reals^n$, denote by $\gd\equiv\gd_*$ that $\gd(t)=\gd_*$ for all $t\in\nninteg$. $\C{n}$ is the class of $n$-times continuously differentiable functions. A function $V:\reals^n\to\reals$ belongs to the class $\posClass{\stEq}$, if it is positive definite and decrescent w.r.t. $\stEq\in\reals^n$ (see \cite[Definition 3.3]{Scherer2015}). A function $\lyapfunIncr:\reals^n\times\reals^n\to\nnreals$ is in $\posClassI$, if 
$\lyapfunIncr(\cdot,\sto)\in\posClass{\sto}$ for all $\sto\in\reals^n$ and $\lyapfunIncr(\st,\cdot)\in\posClass{\st}$ for all $\st\in\reals^n$. $\norm{\centerdot}$ is the Euclidean (vector) norm. We use $(\star)$ to denote a symmetric term {in a quadratic expression, e.g., $(\star)^\top  Q(a-b) = (a-b)^\top Q(a-b)$ for $Q\in\sym^n$ and $a,b\in\reals^{n}$}. The notation $A\succ 0$ ($A\succeq 0$) indicates that $A\in\sym^n$ is positive (semi-) definite, while $A\prec 0$ ($A\preceq 0$) {denotes a} negative (semi-)definite $A\in\sym^n$. Moreover, $\col(x_1, \dots ,x_n)$ denotes the column vector $[x_1^\top\ \cdots\ x_n^\top]^\top$. 

\section{Equilibrium-Free Stability and Performance}\label{8_sec:shiftanalysis}
\subsection{Nonlinear system}
We consider DT nonlinear dynamical systems given in the form of
\begin{subequations}\label{8_eq:nonlinsys}\label{7_an_eq:nl}
\begin{align}
	\st(t+1) &= \stMap(\st(t),\gd(t));\label{8_eq:nonlinsyssteq}\\
	\gp(t) &= \opMap(\st(t),\gd(t));\label{8_eq:nonlinsysopeq}
\end{align}
\end{subequations}
where $t\in\nninteg$ is the discrete-time, $\st(t)\in\reals^\stSize$ is the state with initial condition $\st(0)=\stIc\in\reals^\stSize$, while $\gd(t)\in\reals^\gdSize$ is the input and $\gp(t)\in\reals^\gpSize$ is the output of the system. Moreover, the functions $\stMap:\reals^\stSize\!\times\reals^\gdSize\to\reals^\stSize$ and $\opMap:\reals^\stSize\!\times\reals^\gdSize\to\reals^\opSize$ are considered to be in $\C{1}$. We define the set of solutions of \cref{8_eq:nonlinsys} as
\begin{equation}\label{8_eq:behavior}
	\B := \lbrace (\st,\gd,\gp)\in(\reals^\stSize\!\times\reals^\gdSize\!\times\reals^\gpSize)^{\nninteg}\mid (\st,\gd,\gp)\text{ satisfy \cref{8_eq:nonlinsys}} \rbrace,
\end{equation}
called the behavior of \cref{8_eq:nonlinsys}. For a specific input $\bar\gd\in(\reals^\gdSize)^{\nninteg}$, %
\begin{equation}
	\Bw(\bar\gd) := \{ (\st,\bar\gd,\gp)\in\B\},
\end{equation}
denotes the compatible solution trajectories of  \cref{8_eq:nonlinsys}.
We also define the state transition map $\sttran:\nninteg\times\nninteg\times\reals^\stSize\!\times(
\reals^\gdSize)^{\nninteg}\to\reals^\stSize$, such that $\st(t) = \sttran(t,0,\stIc,\gd)$. Moreover, we assume that all solutions are forward complete and unique.  

For %
\cref{8_eq:nonlinsys}, the equilibrium points satisfy
\begin{subequations}\label{8_eq:equilibriumeq}
\begin{align}
	\stEq &= \stMap(\stEq,\gdEq);\label{8_eq:eqx}\\
	\gpEq &= \opMap(\stEq,\gdEq);
\end{align}
\end{subequations}
where $\stEq\in\reals^\stSize$, $\gdEq\in\reals^\gdSize$ and $\gpEq\in\reals^\gpSize$. The set of equilibrium points of the nonlinear system is then defined as 
\begin{equation}\label{8_eq:eqsetdef}
	\eqSet\!:=\!\lbrace (\stEq,\gdEq,\gpEq)\in\reals^\stSize\!\times\! \reals^\gdSize\!\times\!\reals^\gpSize\mid  (\stEq,\gdEq,\gpEq)\text{ satisfy \cref{8_eq:equilibriumeq}}\rbrace.
\end{equation}
Define $\stSetEq := \proj_\mr{\stEq}\eqSet$, $\gdSetEq := \proj_\mr{\gdEq}\eqSet$, and $\gpSetEq := \proj_\mr{\gpEq}\eqSet$. For our results concerning universal shifted stability and performance, we take the following assumption:
\begin{assumption}[Unique equilibria]\label{4_assum:uniqueEq}
	For the nonlinear system given by \cref{8_eq:nonlinsys} with equilibrium points $(\stEq,\gdEq,\gpEq)\in\eqSet$, we assume that there exists a bijective map $\eqMap:\gdSetEq\to \stSetEq$ such that $\stEq = \eqMap(\gdEq)$, for all $(\stEq,\gdEq)\in \proj_\mr{\stEq,\gdEq}\eqSet$. This means that, for each $\gdEq\in\gdSetEq$, there is a unique corresponding $\stEq\in\stSetEq$, and vice versa, for each $\stEq\in\stSetEq$ there is a unique corresponding $\gdEq\in\gdSetEq$.
\end{assumption}
Note that this assumption is only taken in order to simplify the discussion.

\subsection{Equilibrium-free stability}
As aforementioned, in this paper, we will consider two forms of equilibrium-free stability. Namely, universal shifted stability and incremental stability. In DT, we consider the same definition for universal shifted stability as has been considered in CT in \cite{Koelewijn2023}, i.e., a system given by \cref{8_eq:nonlinsys} is universally shifted (asymptotically) stable if it is (asymptotically) stable w.r.t. to all its forced equilibrium points. More concretely:
\begin{definition}[Universal shifted stability \cite{Koelewijn2023}]\label{4_def:shiftedstab}
	The nonlinear system given by \cref{8_eq:nonlinsys} is \emph{Universally Shifted Stable} (USS), if for each $\epsilon>0$ and $\stEq \in\stSetEq$ with corresponding $\gdEq\in\gdSetEq$, i.e., $(\stEq,\gdEq)\in \proj_\mr{\stEq,\gdEq}\eqSet$, there exists a $\delta(\epsilon)>0$ s.t. any $\st\in\Bw(\gd\equiv\gdEq)$ with $\norm{\st(0)-\stEq} < \delta(\epsilon)$ satisfies %
	$\norm{\st(t)-\stEq}< \epsilon$ for all $t\in\nninteg$. %
	The system is \emph{Universally Shifted Asymptotically Stable} (USAS) if it is USS and, for each $(\stEq,\gdEq)\in \proj_\mr{\stEq,\gdEq}\eqSet$, there exists a $\delta > 0$ such that for $\gd\equiv\gdEq$ we have that $\norm{\st(0)-\stEq} < \delta(\epsilon)$ implies $\lim_{t\to\infty}\norm{\sttran(t,0,\st(0),\gd)-\stEq}=0$.
\end{definition}

To the authors' knowledge, a Lyapunov based theorem to imply US(A)S for DT systems does not yet exist literature. Therefore, we provide the following novel result:
\begin{theorem}[Universal shifted Lyapunov stability]\label{8_thm:shiftlyapstab}
	The nonlinear system given by \cref{8_eq:nonlinsys} is USS, if there exists a function $\lyapfunShift:\reals^\stSize\!\times\gdSetEq\to\nnreals$ with $\lyapfunShift(\cdot,\gdEq)\in\C{1}$ and $\lyapfunShift(\cdot,\gdEq)\in\posClass{\stEq}$ for every $(\stEq,\gdEq)\in\proj_\mr{\stEq,\gdEq}\eqSet$, such that
		\begin{equation}\label{8_eq:shiftedstability}
		\lyapfunShift(\st(t+1),\gdEq)-\lyapfunShift(\st(t),\gdEq)\leq 0,
	\end{equation}
	holds for every $(\stEq,\gdEq)\in\proj_\mr{\stEq,\gdEq}\eqSet$ and for all $t\in\nninteg$ and $\st\in\proj_\mr{\st}\Bw(\gd\equiv \gdEq)$. If \cref{8_eq:shiftedstability} holds, but with strict inequality except when $\st(t)=\stEq$, then the system is USAS.
\end{theorem}
\begin{proof}
See \cref{8_pf:shiftlyapstab}
\end{proof}

Similarly to how standard Lyapunov stability gives rise to invariant sets around stable equilibrium points of a system, we can also extend the notion of invariance for universal shifted Lyapunov stability: 
\begin{theorem}[Universal shifted invariance]\label{thm:shiftinvar}
	For a nonlinear system given by \cref{8_eq:nonlinsys}, for which there exists a function $\lyapfunShift:\reals^\stSize\!\times\gdSetEq\to\nnreals$ such that it is USS in terms of \cref{8_thm:shiftlyapstab},
	any level set:
	\begin{equation}\label{eq:shiftinvarset}
		\mb{X}_{\gdEq,\gamma}=\{x\in \reals^\stSize \mid \lyapfunShift(x,\gdEq) \leq \gamma \},
	\end{equation}
	with $\gamma\geq0$ is invariant, meaning that 
	\begin{equation}\label{eq:sttranshift}
		\sttran(t,0,\stIc,\gd\equiv\gdEq) \in \mb{X}_{\gdEq,\gamma},
	\end{equation}
	for all $t\in\nninteg$, $\stIc\in\mb{X}_{\gdEq,\gamma}$.
\end{theorem}
\begin{proof}
See \cref{pf:USinvariance}
\end{proof}
Note that this notion of universal shifted invariance can be interpreted as the existence of standard Lyapunov based invariant sets around each (forced) equilibrium point $(\stEq, \gdEq)$ of the system.

Incremental stability is a stronger notion than universal shifted stability and considers that all trajectories should be stable w.r.t. each other, meaning that all trajectories converge towards each other. Therefore, incremental stability also implies universal shifted stability. For incremental stability, various definitions exist in the literature, such as in \cite{Tran2016, Tran2018}. Here, we will consider the following slightly more general definition:
\begin{definition}[Incremental stability]\label{4_def:incrstab}
	The nonlinear system given by \cref{8_eq:nonlinsys} is \emph{Incrementally Stable} (IS), if for each $\epsilon>0$ there exists a $\delta(\epsilon)>0$ s.t. any $\st,\sto\in\Bw(\gd)$ with $\norm{\st(0)-\sto(0)} < \delta(\epsilon)$ satisfies $\norm{\st(t)-\sto(t)}< \epsilon$ for all $t\in\nninteg$. The system is \emph{Incrementally Asymptotically Stable} (IAS) if it is IS and there exists a $\delta > 0$ such that $\norm{\st(0)-\sto(0)} < \delta(\epsilon)$ implies that $\lim_{t\to\infty}\norm{\sttran(t,0,\st(0),\gd)-\sttran(t,0,\sto(0),\gd)}=0$.
\end{definition}
Similar as for US(A)S, a Lyapunov Theorem for I(A)S can also be formulated, which we we adopt from \cite{Tran2016,Tran2018}:
\begin{theorem}[Incremental Lyapunov stability \cite{Tran2016,Tran2018}:]\label{7_an_thm:incrlyap}
	The nonlinear system given by \cref{7_an_eq:nl} is IS, if there exists an incremental Lyapunov function $\lyapfunIncr:\reals^\stSize\!\times\reals^\stSize\!\to\nnreals$ with $\lyapfunIncr\in\C{1}$ and $\lyapfunIncr\in\posClassI$, such that
\begin{equation}\label{7_an_eq:incrlyap}
	\lyapfunIncr(\st(t+1),\sto(t+1))-\lyapfunIncr(\st(t),\sto(t))\leq 0,
\end{equation}
for all $t\in\nninteg$ and $\st,\sto\in\proj_{\mr{\st}}\Bw(\gd)$ under all measurable and bounded $\gd\in(\reals^\gpSize)^{\nninteg}$. Moreover, the nonlinear system is IAS, if \cref{7_an_eq:incrlyap} holds, but with strict inequality except when $\st(t)=\sto(t)$.
\end{theorem}
See \cite{Tran2016,Tran2018} for the proof. Finally, we can also define a notion of invariance for incremental stability:
\begin{theorem}[Incremental invariance]\label{thm:incrinvar}
For a nonlinear system given by \cref{8_eq:nonlinsys}, for which there exists an incremental Lyapunov function $\lyapfunIncr:\reals^\stSize\!\times\reals^\stSize\!\to\nnreals$ such that it is IS, and for any given trajectory $(\sto,\gd)\in\proj_\mr{\st,\gd}\in\Bw(\gd)$ for which $\gd$ is bounded and measurable, a $\gamma\geq0$ defines a time-varying invariant set:
\begin{equation}\label{eq:incrinvarset}
\mb{X}_{\sto,\gamma}(t) = \left\lbrace x \in\reals^\stSize \mid \lyapfunIncr(\st,\sto(t))\leq\gamma\right\rbrace	,
\end{equation}
i.e., an invariant tube, such that 
\begin{equation}\label{eq:sttranincr}
		\sttran(t,0,\stIc,\gd) \in \mb{X}_{\sto,\gamma}(t),
	\end{equation}
for all $t\in\nninteg$ and $\stIc\in\mb{X}_{\sto,\gamma}(0)$.
\end{theorem}
\begin{proof}
See \cref{pf:IncrInvariance}
\end{proof}
In the case of IAS, there exists a $\gamma:\nninteg\to\nnreals$, which is a monotonically decreasing function such that $\mb{X}_{\sto,\gamma}(t) = \left\lbrace x \in\reals^\stSize \mid \lyapfunIncr(\st,\sto(t))\leq\gamma(t)\right\rbrace$.

A visual illustration of the time-varying invariant set $\mb{X}_{\sto,\gamma}$ for incremental invariance is depicted in \cref{fig:incrinvar}. As it is visible in the figure, incremental invariance can be interpreted the existence of an invariant tube around a given trajectory $\sto\in\proj_{\mr{\st}}\Bw(\gd)$.

\begin{figure}
	\centering
	\includegraphics[scale=0.9]{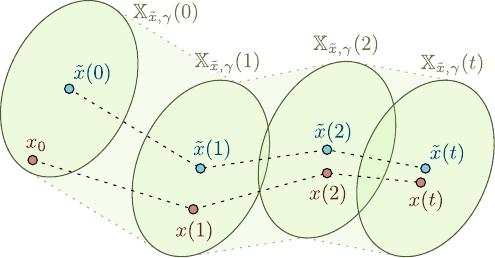}
	\caption{The invariant tube $\mb{X}_{\sto,\gamma}$ for incremental invariance.}
	\label{fig:incrinvar}
\end{figure}

\subsection{Equilibrium-free dissipativity}
The concept of dissipativity \cite{Willems1972} allows for simultaneous analysis of stability and performance of systems. The concept of ``classical'' dissipativity can be interpreted as analyzing the internal energy of the system over time. However, this analysis of internal energy of the system is only concerned with respect to a single ``minimum'' point, called the neutral storage, which is often taken as the origin of the state-space associated with the nonlinear representation. Nevertheless, it is often of interesest to analyze a set of equilibrium points/trajectories, e.g., in the case of reference tracking or disturbance rejection, which is cumbersome to be performed with the classical dissipativity results for nonlinear systems. Hence, there is a need for equilibrium-free dissipativity notions such as universal shifted dissipativity and incremental dissipativity, as they allow to handle these cases efficiently without the restriction of a single point of neutral storage. 

The concept of \emph{Universal Shifted Dissipativity} (USD) \cite{Koelewijn2023} allows for analyzing the energy flow between trajectories and equilibrium points of the system. More concretely, similar to the CT USD notion in \cite[Definition 2]{Koelewijn2023}, we formulate the following definition of DT USD:
\begin{definition}[Universal shifted dissipativity]\label{8_def:shifteddissip}
	The nonlinear system given by \cref{8_eq:nonlinsys} is \emph{Universally Shifted Dissipative} (USD) w.r.t. the supply function $\supfunShift:\reals^\gdSize\!\times\gdSetEq\times \reals^\gpSize\!\times \gpSetEq \to \reals$, if there exists a storage function $\storfunShift:\reals^\stSize\!\times\gdSetEq\to\nnreals$ with $\storfunShift(\cdot,\gdEq)\in\C{0}$  and $\storfunShift(\cdot,\gdEq)\in\posClass{\stEq}$ for every $(\stEq,\gdEq)\in\proj_\mr{\stEq,\gdEq}\eqSet$,  such that
	\begin{equation}\label{8_eq:shifteddissip}
	\storfunShift(\st(t_1+1),\gdEq)-\storfunShift(\st(t_0),\gdEq)\leq \sum_{t=t_0}^{t_1} \supfunShift(\gd(t),\gdEq,\gp(t),\gpEq),
\end{equation}
for all $t_0,t_1\in\nninteg$ with $t_0\leq t_1$ and $(\st,\gd,\gp)\in\B$.
\end{definition}

Incremental dissipativity, see \cite{Verhoek2020}, is an even stronger notion of dissipativity which takes into account multiple trajectories of a system and can be thought of as analyzing the energy flow between trajectories. Similar to the incremental dissipativity definition for CT systems in \cite{Verhoek2020}, we formulate the definition of incremental dissipativity of DT nonlinear systems as follows:
\begin{definition}[Incremental dissipativity]\label{7_an_def:incrdis}
The system given by \cref{8_eq:nonlinsys} is called \emph{Incrementally Dissipative} (ID) w.r.t. the supply function $\supfunIncr:\reals^\gdSize\!\times\reals^\gdSize\!\times\reals^\gpSize\!\times\reals^\gpSize\!\to\reals$, if there exists a storage function ${\storfunIncr}: \reals^\stSize\!\times\reals^\stSize\!\to\nnreals$ with $\storfunIncr\in\C{0}$ and $\storfunIncr\in\posClassI$, such that, for any two trajectories $(\st,\gd,\gp),(\sto,\gdo,\gpo)\in\B$,
\begin{multline}
{\storfunIncr}\big(\st(t_1+1), \sto(t_1+1)\big) - {\storfunIncr}\big(\st(t_0), \sto(t_0)\big) \le \\\sum_{t=t_0}^{t_1}\supfunIncr \big(\gd(t), \gdo(t),\gp(t), \gpo(t)\big),
\label{7_an__eq:IDIE}
\end{multline}
for all $t_0,t_1\in\nninteg$ with $t_0 \leq t_1$. 
\end{definition}

For classical dissipativity, supply functions of a quadratic form are often studied as they allow us to link dissipativity of a system to (quadratic) performance notions such as the \dltwo-gain and passivity. Similarly, for this reason, we will also focus on quadratic supply functions for USD and ID in this paper. More concretely, we will consider quadratic supply functions of the form
\begin{equation}\label{8_eq:shiftsupply}
	\supfunShift(\gd,\gdEq,\gp,\gpEq) = \begin{bmatrix}
		\gd-\gdEq\\ \gp-\gpEq
	\end{bmatrix}^\top \qsrMat\begin{bmatrix}
		\gd-\gdEq\\ \gp-\gpEq
	\end{bmatrix},
\end{equation}
where $\supQ\in\sym^\gdSize$, $\supR\in\sym^\gpSize$, and $\supS \in\reals^{\gdSize\times\gpSize}$. In the incremental case, we choose $\supfunIncr(\gd,\gdo,\gp,\gpo)$ to be defined similarly to \cref{8_eq:shiftsupply},
We will refer to USD and ID w.r.t. supply functions of the form \cref{8_eq:shiftsupply} as $\qsr$-USD and $\qsr$-ID, respectively.

 Like in CT, it can easily be shown that $\qsr$-USD or $\qsr$-ID of a DT nonlinear system given by \cref{8_eq:nonlinsys} with $\qsr = (\perf^2,0,-I)$ implies that the system has a universal shifted or incremental \dltwo-gain  bound of $\perf$, respectively (see \cite[Definition 3]{Koelewijn2023} and \cite{Koelewijn2021a}). Similarly, a DT nonlinear system is universally shifted or incrementally passive if it is $\qsr$-USD or $\qsr$-ID with $\qsr = (0,I,0)$, respectively.
 
The definitions of USD and ID give us conditions to analyze universal shifted and incremental stability and performance properties of DT nonlinear systems. However, using these conditions directly to analyze these notions is difficult, as they require finding a storage function that satisfies the corresponding conditions w.r.t. all equilibria or w.r.t. any solution pairs of the system. Therefore, in the next two sections, we will show how other dissipativity notions can be used to simplify the analysis of USD and ID of DT nonlinear systems.

\section{Velocity Analysis}\label{8_sec:veloanalysis}
\subsection{The DT velocity form and velocity dissipativity}

In this section, we will focus on analyzing US(A)S and \emph{Universal Shifted Performance} (USP) properties of DT nonlinear systems using so-called velocity based analysis. In \cite{Koelewijn2023}, it has been shown how these properties for CT nonlinear systems could be analyzed through the time-differentiated dynamics, i.e., velocity form of the system. In DT, the counterpart to time-differentiation is taking difference of the dynamics in time. Due to the different nature of the difference and derivative operators, the resulting velocity form in DT is different from the CT version. This also results in proofs that are of different nature, than their CT counterparts. As a contribution of this paper, we will show in this section how the time-difference dynamics in DT can be used to imply USS and USP of the original DT nonlinear system. 

Let us introduce the forward increment signals $\dtst(t):=\st(t+1)-\st(t)\in\reals^\stSize$, $\dtgd(t):=\gd(t+1)-\gd(t)\in\reals^\gdSize$, and $\dtgp(t):=\gp(t+1)-\gp(t)\in\reals^\gpSize$, which sometimes are also called as DT velocities.
 Analogously, we can introduce the more commonly used backward increment signals  $\bdtst(t):=\st(t)-\st(t-1)\in\reals^\stSize$, $\dots$, $\bdtgp(t):=\gp(t)-\gp(t-1)\in\reals^\gpSize$.

 Based on these variables, the \emph{forward time-difference dynamics} of \cref{8_eq:nonlinsys} can be expressed as
\begin{subequations}\label{8_eq:veloformFull}
	\begin{align}
			\dtst(t+1) &= \stMap(\st(t+1),\gd(t+1))-\stMap(\st(t),\gd(t));\\
	\dtgp(t) &= \opMap(\st(t+1),\gd(t+1))-\opMap(\st(t),\gd(t));
\end{align}
\end{subequations}
while the \emph{backward time-difference dynamics} of \cref{8_eq:nonlinsys} are
\begin{subequations}\label{9_eq:veloformFull}
	\begin{align}
\bdtst(t+1) &= \stMap(\st(t),\gd(t))-\stMap(\st(t-1),\gd(t-1));\\
	\bdtgp(t) &= \opMap(\st(t),\gd(t))-\opMap(\st(t-1),\gd(t-1)).
\end{align}
\end{subequations}
Let us define the operator $\deltB$ for the behavior $\B$ of \cref{8_eq:nonlinsys} %
such that \vspace{-2mm}
\begin{multline} \label{deltaset}
	\deltB \B = \big\{(\dtst,\dtgd,\dtgp)\in (\reals^\stSize\times \reals^\gdSize\times \reals^\gpSize)^{\nninteg} \mid \\ \dtst(t)=\st(t+1)-\st(t),\,\dtgd(t)=\gd(t+1)-\gd(t), \\\dtgp(t):=\gp(t+1)-\gp(t),\, \forall\, t\in\nninteg,\, (\st,\gd,\gp) \in \B\big\},
\end{multline}
which defines the solution set of the forward dynamics \eqref{8_eq:veloformFull}. If $q$ is the forward-time shift operator, meaning that $qx(t)=x(t+1)$, then $\deltB\B =q \nabla \B$, where $\nabla \B$ is the solution set of the backward dynamics \eqref{9_eq:veloformFull}. As \cref{8_eq:nonlinsys} is time-invariant, this concludes that \eqref{8_eq:veloformFull} are \eqref{9_eq:veloformFull} equivalent representations. For technical convenience, we will formulate our results  w.r.t. the forward dynamics \eqref{8_eq:veloformFull}, but all derivations can be equivalently formulated for \eqref{9_eq:veloformFull} as well. Please also note that both representations \eqref{9_eq:veloformFull} and \eqref{8_eq:veloformFull} are causal.

By the (second) fundamental theorem of calculus, we can equivalently write %
\cref{8_eq:veloformFull} in an alternative form, which we will refer to as the DT \emph{velocity} form:
\begin{definition}[Discrete-time velocity form]\label{8_def:dtveloform}
The velocity form of a nonlinear system, given by \cref{8_eq:nonlinsys} with $f,h\in \C{1}$, is
\begin{subequations}\label{8_eq:veloform}
	\begin{align}
		\dtst(t\!+\!1) \!&=\! \vintA\big(\stip(t\!+\!1),\stip(t)\!\big)\dtst(t) \!+\! \vintB\big(\stip(t\!+\!1),\stip(t)\!\big)\dtgd(t);\\
	\dtgp(t) \!&=\! \vintC\big(\stip(t\!+\!1),\stip(t)\!\big)\dtst(t) \!+\! \vintD\big(\stip(t\!+\!1),\stip(t)\!\big)\dtgd(t);
\end{align}
\end{subequations}
where $(\st,\gd,\gp)\in\B$, $\stip = \col(\st,\gd)$, and
\begin{subequations}\label{8_eq:velointmat}
\begin{align}
\vintA(\stpl,\gdpl,\st,\gd) &= \int_0^1 \Partial{\stMap}{\st}(\stb(\var),\gdb(\var))\,d\var, \\ 
\vintB(\stpl,\gdpl,\st,\gd) &= \int_0^1 \Partial{\stMap}{\gd}(\stb(\var),\gdb(\var))\,d\var, \\
\vintC(\stpl,\gdpl,\st,\gd) &= \int_0^1 \Partial{\opMap}{\st}(\stb(\var),\gdb(\var))\,d\var, \\
\vintD(\stpl,\gdpl,\st,\gd) &= \int_0^1 \Partial{\opMap}{\gd}(\stb(\var),\gdb(\var))\,d\var,
\end{align}
\end{subequations}
with $\stb(\var) = \st+\var(\st_+-\st)$, $\gdb(\var) =\gd+\var(\gd_+-\gd)$.
\end{definition}

Note that by the \emph{Fundamental Theorem of Calculus} \cite{Thomas2005}, \eqref{8_eq:veloformFull} and \eqref{8_eq:veloform} are equivalent. To distinguish the velocity form of \cref{8_eq:nonlinsys}  from the original nonlinear system, we will call \cref{8_eq:nonlinsys} to be the \emph{primal form}.

Based on \cref{deltaset}, the solution set of \cref{8_eq:veloform} is given by $\Bv:=\deltB \B$, and we can also define $\Bvw(\gd) := \deltB \Bw(\gd)$ for a $\gd \in \gdSet^{\nninteg}$. 
The resulting DT velocity form represents the dynamics of the change between consecutive time-instances of the original dynamics. This is analogous to the CT velocity form introduced in \cite{Koelewijn2023}, which %
represents the dynamics of the instantaneous change in time (i.e., time derivative) of the original dynamics. %
Next, we will show that the DT velocity form has a direct relation to USS and USP. Before presenting this connection, we will first show some analysis results %
on the DT velocity form.

\begin{definition}[Velocity stability]\label{def:velostab}
The nonlinear system given by \cref{8_eq:nonlinsys} with velocity form \cref{8_eq:veloform} is \emph{Velocity (Asymptotically) Stable} (V(A)S), if the velocity form is (asymptotically) stable in the Lyapunov sense w.r.t. the origin (see also \cref{4_def:shiftedstab}), i.e., the velocity state $\dtst$ is (asymptotically) stable w.r.t. 0.
\end{definition}

As V(A)S is nothing more than (asymptotic) stability of velocity form, we can easily formulate the following Lyapunov based theorem in order to verify it:
\begin{theorem}[Velocity Lyapunov stability]\label{thm:velostab}
	The nonlinear system given by \cref{8_eq:nonlinsys} is VS, if there exists a function $\lyapfunVelo:\reals^{\stSize}\to\nnreals$ with $\lyapfunVelo\in\C{0}$ and $\lyapfunVelo\in\posClass{0}$, such that
	\begin{equation}\label{8_eq:velostability}
		\lyapfunVelo(\dtst(t+1))-\lyapfunVelo(\dtst(t))\leq 0,
	\end{equation}
	for all $t\in\nninteg$ and $\dtst\in\proj_\mr{\dtst}\Bvset{\gdSetEq}$. If \cref{8_eq:velostability} holds, but with strict inequality except when $\dtst(t)=0$, then the system is VAS.
\end{theorem}
The proof of \cref{thm:velostab} simply follows from standard Lyapunov stability theory, see e.g., \cite{Khalil2002,Bof2018}. Next, we formulate a notion of dissipativity regarding the velocity form, which enables the analysis of stability and performance of nonlinear systems in the velocity sense:
\begin{definition}[Velocity dissipativity]\label{8_def:velodissip}
	The nonlinear system given by \cref{8_eq:nonlinsys} is \emph{Velocity Dissipative} (VD) w.r.t. the supply function $\supfunVelo:\reals^\gdSize\times\reals^\gpSize\to\reals$, if there exists a storage function $\storfunVelo:\reals^\stSize\to\nnreals$ with $\storfunVelo\in\C{0}$ and $\storfunVelo\in\posClass{0}$, such that, for all $t_0,t_1\in\nninteg$ with $t_0\leq t_1$,
\begin{equation}\label{8_eq:velodissip}
	\storfunVelo(\dtst(t_1+1))-\storfunVelo(\dtst(t_0))\leq \sum_{t=t_0}^{t_1} \supfunVelo(\dtgd(t),\dtgp(t)),
\end{equation}
for all $(\dtst,\dtgd,\dtgp)\in\Bv$.
\end{definition}
Note that VD can be seen as `classical' dissipativity of the velocity form of the system \cref{8_eq:veloform}. Next, let us consider quadratic $\qsr$ supply functions for VD of the form
\begin{equation}\label{8_eq:velosupply}
	\supfunVelo(\dtgd,\dtgp) = \begin{bmatrix}
		\dtgd\\ \dtgp
	\end{bmatrix}^\top\qsrMat\begin{bmatrix}
		\dtgd\\ \dtgp
	\end{bmatrix},
\end{equation}
where again $\supQ\in\sym^\gdSize$, $\supS \in\reals^{\gdSize\times\gpSize}$, and $\supR\in\sym^\gpSize$. Moreover, we also consider the storage function $\storfunVelo$ to be quadratic:
\begin{equation}\label{8_eq:velostorquad}
	\storfunVelo(\dtst) = \dtst^\top \storquad \dtst ,
\end{equation}
where $\storquad\in\sym^\stSize$ with $\storquad \succ 0$. Under these considerations, we can derive the following (infinite dimensional) \emph{Linear Matrix Inequality} (LMI) feasibility condition for VD:
\begin{theorem}[DT $\qsr$-VD condition]\label{8_thm:veloqsrMI}
	The system given by \cref{8_eq:nonlinsys} is $\qsr$-VD on the convex set $\stSet\times\gdSet \subseteq \reals^{\stSize\cdot\gdSize}$, where $\supR\preceq 0$, if there exists an $\storquad\in\sym^\stSize$ with $\storquad \succ 0$, such that for all $(\st,\gd)\in\stSet\times\gdSet$, it holds that
\begin{multline}\label{8_eq:veloMI}
(\star)^\top  \begin{bmatrix} -\storquad & 0 \\\star &\storquad \end{bmatrix}  
\begin{bmatrix} 
I & 0 \\ 
\velA(\st,\gd) & \velB(\st,\gd)
\end{bmatrix} 
-\\(\star)^\top \qsrMat \begin{bmatrix} 
0 & I \\ 
\velC(\st,\gd) & \velD(\st,\gd)
\end{bmatrix}  \preceq 0,
\end{multline}
where $\velA=\Partial{\stMap}{\st}$, $\velB=\Partial{\stMap}{\gd}$, $\velC=\Partial{\opMap}{\st}$, $\velD=\Partial{\opMap}{\gd}$.
\end{theorem}
\begin{proof}
See \cref{8_pf:veloqsrMI}.
\end{proof}

\begin{remark} \label{rem:1}
	When we talk about a system being stable or dissipative on (a set) $\stSet\times\gdSet$, we mean the system is stable or dissipative under all trajectories of the system for which holds that $(\st(t),\gd(t))\in\stSet\times\gdSet$ for all $t\in\nninteg$.
\end{remark}

Later, in \cref{sec:diffanalysis}, we will also show how, for performance in terms of the $\dltwo$-gain and passivity, the results of \cref{8_thm:veloqsrMI} can be turned into LMIs. With the result of \cref{8_thm:veloqsrMI}, we have a novel condition to analyze velocity dissipativity of DT nonlinear systems. This is enabled by the fact that, in the proof, it is shown that the condition for $\qsr$-VD in DT can be expressed in terms of the matrix functions $\velA,\dots,\velD$ instead of the matrix functions $\vintA,\dots,\vintD$ of the DT velocity form \cref{8_eq:veloform}. Expressing the condition for VD in terms of $\velA,\dots,\velD$ instead of $\vintA,\dots,\vintD$ simplifies it. Namely, $\velA,\dots,\velD$ only depend on two arguments, which results in the condition needing to be verified at all $(\st,\gd)\in\stSet\times\gdSet$. On the other hand, a condition using $\vintA,\dots,\vintD$ takes four arguments and would need to be verified at all $(\st,\gd)\in\stSet\times\gdSet$ \emph{and} all $(\stpl,\gdpl)\in\stSet\times\gdSet$. 

Moreover, note that the condition in \cref{8_thm:veloqsrMI} corresponds to a feasibility check of an infinite dimensional set of LMIs, as for a fixed  $(\st,\gd)\in\stSet\times\gdSet$, \cref{8_eq:veloMI} becomes an LMI. Later, in \cref{8_sec:lpvstuffs}, we will see how we can reduce this infinite dimensional set of LMIs to a finite dimensional set, which can computationally efficiently be verified. This will then give us %
efficient tools to analyze $\qsr$-VD of a system.

\subsection{Induced universal shifted stability}\label{8_sec:indshiftstab}
In the literature, see \cite{Koelewijn2023, Kosaraju2019,Kawano2021}, it has been shown how the velocity form in CT can be used to formulate a condition to imply US(A)S of CT systems. Likewise, we will show that also in DT, we can formulate a condition for US(A)S of a system using the DT velocity form that we have introduced in \cref{8_def:dtveloform}. Before doing so, let us first introduce the behavior $\Bvset{\gdSetEq}:=\bigcup_{\gdEq\in\gdSetEq} \Bvw( \gd \equiv \gdEq)$, i.e., the behavior of the velocity form for which the input is $\gd(t)=\gdEq\in\gdSetEq$, hence $\dtgd(t)=0$, for all $t\in\nninteg$.
\begin{theorem}[Implied universal shifted stability]\label{8_thm:velotoshiftstab}
	The nonlinear system given by \cref{8_eq:nonlinsys} is USS, if there exists a function $\lyapfunVelo:\reals^{\stSize}\to\nnreals$ with $\lyapfunVelo\in\C{0}$ and $\lyapfunVelo\in\posClass{0}$, such that \cref{8_eq:velostability} holds for all $t\in\nninteg$ and $\dtst\in\proj_\mr{\dtst}\Bvset{\gdSetEq}$, i.e., the system is velocity stable. If \cref{8_eq:velostability} holds, but with strict inequality except when $\dtst(t)=0$, meaning it is velocity asymptotically stable, then the system is USAS.
\end{theorem}
\begin{proof}
See \cref{8_pf:velotoshiftstab}.
\end{proof}
The proof for \cref{8_thm:velotoshiftstab} relies on the construction of the universally shifted Lyapunov function based on $\lyapfunVelo$. In CT, a similar construction is often referred to as the Krasovskii method \cite{Khalil2002,Kawano2021,Koelewijn2023}. However, the novel result and construction that we present in \cref{8_thm:velotoshiftstab} for DT nonlinear systems are, to the authors' knowledge, not available in literature. Moreover, to the authors' knowledge, this is also the first time that properties of the time-difference dynamics have been connected to US(A)S of the system. 
Note that %
condition \cref{8_eq:velostability} means that (asymptotic) stability of the velocity form \eqref{8_def:dtveloform} implies US(A)S of system \cref{8_eq:nonlinsys}. Which implies that by analyzing (asymptotic) stability of the velocity form \cref{8_eq:veloform}, we can infer US(A)S of the primal form.

Using \cref{8_thm:velotoshiftstab}, we can also connect velocity dissipativity to US(A)S of the nonlinear system:
\begin{theorem}[USS from VD]\label{4_lem:velostab}
	Assume the nonlinear system given by \cref{8_eq:nonlinsys} is VD under a storage function $\storfunVelo\in\C{1}$ w.r.t. a supply function $\supfunVelo$ that satisfies
	\begin{equation}\label{8_eq:supplystability}
		\supfunVelo(0,\gpdot)\leq 0,
	\end{equation}	
	for all $\gpdot\in\reals^\gpSize$, then, the nonlinear system is USS. If the supply function satisfies \cref{8_eq:supplystability}, but with strict inequality when $\dtst\neq 0$, then the nonlinear system is USAS.
\end{theorem}
\begin{proof}
	See \cref{4_pf:velostab}.
\end{proof}

\begin{corollary}[VD-condition induced universal shifted stability]\label{cor:shiftinvar}
	For the nonlinear system given by \cref{8_eq:nonlinsys}, let \cref{8_eq:veloMI} hold on the convex set $\stSet\times\gdSet \subseteq \reals^{\stSize\cdot\gdSize}$ w.r.t. a supply function $\supfunVelo$ that satisfies \cref{8_eq:supplystability}, i.e., the system is VD and V(A)S on $\stSet\times\gdSet$. Then, for any input $\gd\equiv\gdEq\in\gdSet$, the system is US(A)S and invariant on $\mb{X}_{\gdEq,\gamma}$ given by \cref{eq:shiftinvarset}  with $\lyapfunShift(\st(t),\gdEq)=\storfunVelo(\stMap(\st(t),\gdEq)-\st(t))$, if $\gamma\geq 0$ satisfies $\mb{X}_{\gdEq,\gamma}\subseteq\stSet$.  	

\end{corollary}
The proof simply follows the fact that the system is US(A)S by \cref{4_lem:velostab}, which by \cref{8_thm:velotoshiftstab} implies the system is US(A)S w.r.t. the (universal shifted) Lyapunov function $\lyapfunShift(\st(t),\gdEq)=\storfunVelo(\stMap(\st(t),\gdEq)-\st(t))$. By \cref{thm:shiftinvar}, this then implies universal shifted invariance for any input $\gd\equiv\gdEq\in\gdSet$. Note that \cref{cor:shiftinvar} means that verification of \cref{8_eq:veloMI} on a convex set $\stSet\times\gdSet \subseteq \reals^{\stSize\cdot\gdSize}$ only implies universal shifted  stability on the maximum invariant set  $\mb{X}_{\gdEq,\gamma}$, constructed based on the function $\lyapfunShift$ assembled from $\storfunVelo$, which is still contained in $\stSet$. This is due to the fact that we can only give guarantees for $(\st(t),\gd(t))\in\stSet\times\gdSet$ for all $t\in\nninteg$, as also stated in \cref{rem:1}. For initial conditions in $\stSet \setminus \mb{X}_{\gdEq,\gamma}$, there is no guarantee that the state trajectory will not leave $\stSet$ momentarily and take values where $\cref{8_eq:veloMI}$ has not been verified. Increasing the sets $\stSet\times\gdSet \subseteq \reals^{\stSize\cdot\gdSize}$ allows one to conclude US(A)S on larger regions of the state space.

With these results, we have shown so far that velocity stability and dissipativity imply universal shifted stability and invariance of nonlinear systems. 

\subsection{Induced universal shifted dissipativity}\label{8_sec:veloshiftdissip}
Next, we are interested if $\qsr$-VD also implies $\qsr$-USD. In \cite{Koelewijn2023}, this is also investigated for CT nonlinear systems. However, a full proof for the implication that $\qsr$-VD implies $\qsr$-USD  is not presented and to the authors' knowledge does not exist in the literature for either CT or DT. In this section, we will present novel dual DT conditions that link $\qsr$ velocity dissipativity and $\qsr$ universal shifted dissipativity.

Instead of considering nonlinear systems that can be represented in the form of \cref{8_eq:nonlinsys}, in this section, for technical reasons, we will restrict ourselves to nonlinear systems that can be represented as
\begin{subequations}\label{8_eq:nonlinsysState}
\begin{align}
	\st(t+1) &= \stMap(\st(t))+\ltiB \gd(t);\\
	\op(t) &= \ltiC \st(t).
\end{align}
\end{subequations}
For a system represented by \cref{8_eq:nonlinsys}, we can transform \cref{8_eq:nonlinsys} to the form \cref{8_eq:nonlinsysState} at the cost of increasing the state dimension, e.g., by using appropriate input and output filters (see e.g. \cite[Appendix II]{Koelewijn2023}). For \cref{8_eq:nonlinsysState}, we will also assume in this section that $\st(t)\in\stSet$, with $\stSet$ being convex and compact.

For a nonlinear system given by \cref{8_eq:nonlinsysState}, the equilibrium points $(\stEq,\gdEq,\gpEq)\in\eqSet$ satisfy
\begin{subequations}\label{8_eq:nonlinsysStateEqui}
\begin{align}
	\stEq &= \stMap(\stEq)+\ltiB \gdEq;\\
	\opEq &= \ltiC \stEq;
\end{align}
\end{subequations}
and the velocity form of \cref{8_eq:nonlinsysState} is given by
\begin{subequations}\label{8_eq:nonlinsysStateVelo}
	\begin{align}
		\dtst(t) &= \vintA(\st(t+1),\st(t))\dtst(t)+\ltiB\dtgd(t);\\
	\dtgp(t) &= \ltiC\dtst(t);\label{8_eq:nlsmalloutput}
\end{align}
\end{subequations}
for which $\vintA(\stpl,\st) = \int_0^1 \Partial{\stMap}{\st}(\st+\var(\stpl-\st))\, d\var$.

We will next connect $\qsr$-VD for $\qsr$ tuples %
with $\supS=0$, $\supQ\succeq 0$, and $\supR\preceq 0$ to USP notions that can be characterized by a similar $\qsr$ universal shifted supply function. We take the following assumptions:
\begin{assumption}\label{4_as:CB}
	For the nonlinear system given by \cref{8_eq:nonlinsysState}, assume that $\ltiC \ltiB=0$.
\end{assumption}
While \cref{4_as:CB} may seem restrictive, it can relatively easily be satisfied by interconnecting low pass filters to the inputs and outputs of the system, e.g., see \cite[Appendix II]{Koelewijn2023}. Furthermore, we take the following  commonly used assumption in literature \cite{Koroglu2007, Wieland2009}, namely that  the (generalized) disturbances are generated by a stable exosystem:

\begin{assumption}\label{8_as:wExoSys}
	For a given $(\stEq,\gdEq,\gpEq)\in\eqSet$ and $\beta \in \nnreals$, assume that $\gd$ is generated by the exosystem
	\begin{equation}\label{8_eq:exoSys}
		\gd(t+1) = \exoA(\gd(t)-\gdEq)+\gdEq,
	\end{equation}
	where $\exoA\in\reals^{\gdSize\times\gdSize}$ is Schur and $\norm{\exoA-I}\leq \beta$. The corresponding signal behavior is
	\begin{equation}
		\exoBvr_{(\gdEq,\beta)} :=\left\{ \gd\in \gdSet^{\nninteg} \mid \gd\text{ satisfies } \cref{8_eq:exoSys}\right\}.
	\end{equation}
\end{assumption}

Before presenting our results, we first give the following technical proposition: 
\begin{proposition}\label{8_as:veloShiftBound2}
	Given a matrix $R\in\sym^\gpSize$ with $R\preceq 0$, then there exists an $\alpha\in\preals$, such that for all $\stEq\in\stSetEq$ and $\st\in\stSet$
	\begin{equation}
	\!\!(\star)\!^\top \!R \ltiC\!\left(\vintA(\st,\stEq)\!-\!I\right)\!(\st-\stEq)\leq\alpha^{-1}(\star)\!^\top \!R \ltiC(\st-\stEq).
	\end{equation}
\end{proposition}
In case that $\vintA$ is bounded, there always exists an $\alpha$ for a given $R$ such that the condition in \cref{8_as:veloShiftBound2} holds, as $\stSet$ is compact.

Under \cref{4_as:CB,8_as:wExoSys}, we can show the following result:
\begin{theorem}[USP from VD]\label{8_thm:veloshiftperf}
	If a nonlinear system given by \cref{8_eq:nonlinsysState} is $\qsr$-VD with $\supS=0$, $\supQ\succeq 0$, $\supR\preceq 0$, and $\supR$ satisfies the condition in \cref{8_as:veloShiftBound2}, then under \cref{4_as:CB,8_as:wExoSys},
	for every $(\stEq,\gdEq,\gpEq)\in\eqSet$, it holds that
	\begin{equation}\label{8_eq:pf:qrvspdt}
	\sum_{t=0}^{T} \beta^2(\star)\!^\top \! \supQ (\gd(t)-\gdEq)+\alpha^{-1}(\star)\!^\top \!  \supR (\gp(t)-\gpEq)\geq 0,
\end{equation}
for all $T\geq0$ and $(\gd,\gp)\in\proj_\mr{\gd,\gp}\B$ with $\gd\in\exoBvr_{(\gdEq,\beta)}$ and\footnote{\label{footnote}The results can also be extended to $\dtst(0)\neq 0$, which will introduce an additional constant positive term on the left-hand side of \cref{8_eq:pf:qrvspdt}.} $\dtst(0)=0$.
\end{theorem}
\begin{proof}
See \cref{8_pf:veloshiftperf}.
\end{proof}
Applying the result of \cref{8_thm:veloshiftperf} to the $\qsr$ tuple $\qsr = (\perf^2I,0,-I)$, corresponding to (universal shifted) \dltwo-gain, we obtain the following \lcnamecref{8_cor:veloshiftl2}:
\begin{corollary}[Bounded \dlstwo-gain from velocity dissipativity]\label{8_cor:veloshiftl2}
		If a nonlinear system given by \cref{8_eq:nonlinsysState} is velocity $\qsr$ dissipative for $\qsr=(\perf^2I,0,-I)$,  where $\supR=-I$ satisfies \cref{8_as:veloShiftBound2}, then under \cref{4_as:CB,8_as:wExoSys},
		 the system has an \dlstwo-gain bound of $\tilde\perf = \sqrt{\alpha\beta^2\perf^2}$.
\end{corollary}
Note that this result follows from \cref{8_thm:veloshiftperf} by multiplying \cref{8_eq:pf:qrvspdt} by $\alpha$. The resulting inequality then corresponds to a $\qsr$ US supply function with $\qsr = (\alpha\beta^2\perf^2 I, 0, -I)$, which corresponds to an \dlstwo-gain of $\tilde\perf = \sqrt{\alpha\beta^2\perf^2}$. 

Combining these results with the result of \cref{8_thm:veloqsrMI} gives us a condition to analyze universal shifted performance of DT nonlinear systems in terms of an infinite dimensional set of LMIs given by \cref{8_eq:veloMI} on a chosen convex set $\stSet\times\gdSet \subseteq \reals^{\stSize\cdot\gdSize}$. Through \cref{cor:shiftinvar}, we can clearly characterize for any generalized disturbance  $\gd(t) \in \gdSet$ the region of the state space $\stSet$ where universal shifted stability and performance are guaranteed. As aforementioned, in \cref{7_an_sec:lpv}, we will discuss how can we turn the infinite dimensional set of LMIs into a finite dimensional set in order to cast the analysis problem as convex optimization problem.

The results that we have presented in this section on the connection between universal shifted stability and performance and velocity analysis for DT systems can be seen as the dual of the CT results that have been presented in \cite{Koelewijn2023}. While the results in DT that we have presented in this paper are analogous to the CT results in \cite{Koelewijn2023}, the proofs of the underlying results are very much different due to different nature of the time operators and the velocity forms in CT and DT.

Next, we will show how a different, but similar, dissipativity notion can be used to analyze incremental stability and performance of DT nonlinear systems.

\section{Differential Analysis}\label{sec:diffanalysis}
\subsection{The differential form}

For CT systems, it has been show in \cite{Verhoek2020} how dissipativity of the differential form implies incremental dissipativity. Similarly, in \cite{Koelewijn2021a}, preliminary results have also shown this for DT systems, however, under a restricted form of the storage function. In this section, we provide a novel generalization of these results to show how differential dissipativity implies incremental dissipativity under a state-dependent storage function.

Let us first introduce the following notation: $\pathSet(\varphi,\otherTraj\varphi)$denotes the set of (smooth) paths between points $\varphi,\otherTraj\varphi\in\reals^n$, i.e.,
\begin{equation}\label{5_eq:pathsetdef}
	\pathSet(\varphi,\otherTraj \varphi) := \lbrace \parTraj{\varphi} \in (\reals^n)^{[0,\,1]}\mid \bar \varphi\in\C{1},\, \parTraj{\varphi}(0) = \otherTraj \varphi,\, \parTraj{\varphi}(1)=\varphi\rbrace.
\end{equation}

Next, consider two arbitrary trajectories of the system \cref{8_eq:nonlinsys}: $(\st,\gd,\gp)$, $(\sto,\gdo,\gpo)\in\B$. We parameterize any two trajectories between these in terms of a path connecting their initial conditions: $\stbIc\in\pathSet(\stIc,\stoIc)$ and a path connecting their input trajectories: $\gdb(t)\in\pathSet(\gd(t),\gdo(t))$, resulting in the state transition map $\stb(t,\var)=\sttran(t,t_0,\stbIc(\var), \gdb(\var))\in\reals^\stSize$.
This gives that for any $\var\in[0,\,1]$ and all $(\st,\gd,\gp),(\sto,\gdo,\gpo)\in\B$, it holds that 
\begin{subequations}\label{7_an_eq:nlparam}
\begin{align}
	\stb(t+1,\var) &= \stMap(\stb(t,\var),\gdb(t,\var));\\
	\gpb(t,\var) &= \opMap(\stb(t,\var),\gdb(t,\var));
\end{align}
\end{subequations}
where $(\stb(\var),\gdb(\var),\gpb(\var))\in\B$. Note that for $\var=0$, we obtain $(\stb(0),\gdb(0),\gpb(0)) = (\sto, \gdo,\gpo)\in\B$, while for $\var=1$, we get $(\stb(1),\gdb(1),\gpb(1)) = (\st, \gd,\gp)\in\B$. Differentiating the parameterized dynamics w.r.t. $\var$, results in the so-called differential form of \cref{8_eq:nonlinsys}, given by
\begin{subequations}\label{7_an_eq:sys_diff}
\begin{align}
	\dst (t+1)\!&=\!\difA(\stb(t),\gdb(t))\dst(t) \!+\! \difB(\stb(t),\gdb(t))\dgd(t);\\ 
	\dgp(t)\! &=\! \difC(\stb(t),\gdb(t))\dst(t)  \!+\! \difD(\stb(t),\gdb(t))\dgd(t);
\end{align}
\end{subequations}
where we omitted dependency on $\var$ for the sake of readability. In \cref{7_an_eq:sys_diff}, $\dst(t,\var)= \Partial{\stb}{\var}(t,\var) \in\reals^\stSize$, $\dgd(t,\var)= \Partial{\gdb}{\var}(t,\var) \in\reals^\gdSize$,\linebreak $\dgp(t,\var)= \Partial{\gpb}{\var}(t,\var) \in\reals^\gpSize$, and 
\begin{equation}\label{7_an_eq:diffabcd}
\difA = \Partial{\stMap}{\st},  \quad \difB = \Partial{\stMap}{\gd}, \quad \difC = \Partial{\opMap}{\st},  \quad \difD = \Partial{\opMap}{\gd},
\end{equation}
where $(\stb(\var),\gdb(\var))\in\proj_\mr{\st,\gd}\B$ for all $\var\in[0,\,1]$. The differential form represents the dynamics of the variations along the trajectories of the system represented by \cref{8_eq:nonlinsys}. 

The differential form allows us to define differential stability:
\begin{definition}[Differential stability]\label{def:difstab}
	The nonlinear system given by \cref{8_eq:nonlinsys} with differential form \cref{7_an_eq:sys_diff} is \emph{Differentially (Asymptotically) Stable} (D(A)S), if the differential form is (asymptotically) stable in the Lyapunov sense w.r.t. the origin (see also \cref{4_def:shiftedstab}), i.e., the differential state $\dst$ is (asymptotically) stable w.r.t. 0.
\end{definition}
Similar to the definition of velocity stability in \cref{def:velostab}, differential stability considers standard stability of the differential form. Results for this have been discussed in \cite{Tran2016,Tran2018}, which we will briefly recap: 
\begin{theorem}[Differential Lyapunov stability \cite{Tran2016,Tran2018}]\label{thm:difstab}
	The nonlinear system given by \cref{8_eq:nonlinsys} is DS, if there exists a function $\lyapfunDif:\reals^{\stSize}\times\reals^{\stSize}\to\nnreals$ with $\lyapfunVelo\in\C{0}$ and $\lyapfunDif(\stb,\cdot)\in\posClass{0},\forall\stb\in \reals^\stSize$, such that
	\begin{equation}\label{8_eq:difstability}
		\lyapfunDif(\stb(t+1),\dst(t+1))-\lyapfunDif(\stb(t),\dst(t))\leq 0,
	\end{equation}
	for all $t\in\nninteg$ and for all $\stb\in\proj_{\st}\Bw(\gd)$ under all measurable and bounded $\gd\in(\reals^\gdSize)^\nninteg$. If \cref{8_eq:difstability} holds, but with strict inequality except when $\dst(t)=0$, then the system is DAS.
\end{theorem} 
See also \cite{Tran2016,Tran2018} for the proof. Similarly, using the differential form, we formulate the definition of differential dissipativity, which so far has received little attention in literature in the DT setting:
\begin{definition}[Differential dissipativity]
Consider the system given by \cref{8_eq:nonlinsys} and its differential form  \cref{7_an_eq:sys_diff}. The system is \emph{Differentially Dissipative} (DD) w.r.t. a supply function $\supfunDif:\reals^\gdSize\times\reals^\gpSize\to\reals$, if there exists a storage function $\storfunDif: \reals^\stSize\!\times\reals^\stSize\to\nnreals$ with $\storfunDif\in\C{0}$ and $\storfunDif(\stb,\cdot)\in\posClass{0},\,\forall\,\stb\in\reals^\stSize$, such that
\begin{multline}\label{7_an_eq:DIE_diff}
{\storfunDif}\big(\stb(t_1+1),\dst(t_1+1)\big) - \storfunDif\big(\stb(t_0),\dst(t_0)\big) \leq\\ \sum_{t=t_0}^{t_1}\supfunDif\big(\dgd(t),\dgp(t)\big), 
\end{multline}
for all $(\stb,\gdb)\in\proj_\mr{x,u}\B$ and for all $t_0,t_1\in\nninteg$, with $t_0 \le t_1$.
\label{7_an_eq:diffdiss}
\end{definition}

As in \cref{8_sec:veloanalysis}, we consider quadratic supply functions of the form 
\begin{equation}\label{7_an_eq:supply_diff}
	\supfunDif(\dgd,\, \dgp) = \begin{bmatrix}\dgd \\ \dgp \end{bmatrix}^\top \qsrMat \begin{bmatrix}\dgd\\ \dgp\end{bmatrix}.
\end{equation}
Differential dissipativity w.r.t. supply functions of this form will be referred to as $\qsr$-DD.

Note that checking $\qsr$-DD of the (primal form of the) system \cref{8_eq:nonlinsys} can be seen as checking ``classical $\qsr$ dissipativity'' of the differential form of the system.

For CT systems in \cite{Verhoek2020}, it has been show how $\qsr$-DD can be analyzed through feasibility check of a(n) (infinite dimensional) set of LMIs. In DT, this has also been shown in \cite{Koelewijn2021a}. However, the work in \cite{Koelewijn2021a} only considers a quadratic (differential) storage function with a constant matrix. As a contribution of this paper, we will show that how to formulate a similar result using a quadratic storage function with a state-dependent matrix, and importantly how this condition will imply incremental dissipativity of the nonlinear system. Due to the state-dependent nature of the (differential) storage function, showing this implication is more involved. As mentioned, to obtain the results, we consider storage functions of a quadratic form:  %
	\begin{equation}\label{7_an_eq:diffStorage}
		\storfunDif(\stb,\dst) = \dst^\top \storquad(\stb) \dst ,
	\end{equation}
	with $\storquad$ satisfying the following condition:
\begin{condition}\label{5_as:mbar}
	The matrix function $\storquad:\reals^\stSize\to\sym^\stSize$ with $\storquad\in\C{1}$ is real, symmetric, bounded and positive definite, i.e., $\exists \, \slackA_1, \slackA_2 \in\preals$, such that for all $\stb \in\reals^\stSize$, $\slackA_1 I\preceq \storquad(\stb ) \preceq \slackA_2 I$.
\end{condition}
	 Moreover, for the system  \cref{8_eq:nonlinsys}, let us also consider the set $\stdotSet$, which is the smallest convex set such that $(\st(t+1)-\st(t))\in\stdotSet$ for all $t \in\nninteg$ corresponding to a given set $\stSet\times\gdSet\subseteq\reals^{\stSize\cdot\gdSize}$. This allows us to obtain the following result to analyze differential dissipativity in DT:
\begin{theorem}[$\qsr$-DD condition]\label{7_an_thm:diffdissip}
	A nonlinear system given by \cref{8_eq:nonlinsys} is $\qsr$-DD on $\stSet\times\gdSet \subseteq \reals^{\stSize \cdot \gdSize}$, if there exists a storage function 
	\cref{7_an_eq:diffStorage} with $\storquad$ satisfying \cref{5_as:mbar}, such that 	\begin{multline}\label{7_an_eq:diffDissipFull}
	(\star)^\top \begin{bmatrix}
		-\storquad(\stb) & {0}\\{0} & \storquad(\stb+\stbdot)
	\end{bmatrix}\begin{bmatrix}
		I & {0}\\\difA(\stb,\gdb) & \difB(\stb,\gdb)
	\end{bmatrix}-\\(\star)^\top \qsrMat \begin{bmatrix}
		{0} & I\\\difC(\stb,\gdb) & \difD(\stb,\gdb)
	\end{bmatrix}\preceq 0,
	\end{multline}
	for all $(\stb,\gdb)\in{\stSet}\times{\gdSet}$ and $\stbdot\in\stdotSet$.
\end{theorem}
\begin{proof}
See \cref{7_an_pf:diffdissip}.
\end{proof}

Note that \cref{7_an_eq:diffDissipFull} is similar to the check for $\qsr$-VD in \cref{8_eq:veloMI} in \cref{8_thm:veloqsrMI}. We will discuss the connections and similarities between differential and velocity dissipativity in more detail in \cref{sec:veloanddiff}.

Later, in \cref{7_an_sec:lpv}, we will discuss how we can formulate computable tests for checking feasibility of this infinite dimensional set of LMIs. 

\subsection{Induced incremental stability}

Before showing how $\qsr$ differential dissipativity implies $\qsr$ incremental dissipativity, we will first briefly discuss how it connects to I(A)S.

In literature, the connections between differential dynamics and incremental stability have extensively been discussed, also for DT nonlinear systems \cite{Tran2016,Tran2018}. For completeness, we will briefly recap these results:
\begin{lemma}[Implied incremental stability]\label{7_an_lem:diffstabcondmi}
	The nonlinear system given by \cref{8_eq:nonlinsys} is incrementally stable, if there exists a quadratic storage Lyapunov function $\lyapfunDif$ of the form \cref{7_an_eq:diffStorage} with $M$ satisfying \cref{5_as:mbar}, such that
	\begin{equation}\label{7_an_eq:diffstabmi}
		\lyapfunDif(\stb(t+1), \dst(t+1)) -  \lyapfunDif(\stb(t), \dst(t)) \leq 0
	\end{equation}
	for all $\stb\in\proj_{\st}\Bw(\gd)$ under all measurable and bounded $\gd\in(\reals^\gdSize)^\nninteg$. If \cref{7_an_eq:diffstabmi} holds, but with strict inequality except when $\st(t)=\sto(t)$, corresponding to $\dst(t)=0$, then the system is incrementally asymptotically stable.
\end{lemma}	
Similar to the implication of US(A)S from velocity dissipativity, see \cref{4_lem:velostab}, we also have that I(A)S is implied form differential dissipativity under restrictions of the supply function, see also \cite[Remark 11]{Verhoek2020} for these results in CT.
\begin{theorem}[IS from DD]\label{4_lem:difstab}
	Assume the nonlinear system given by \cref{8_eq:nonlinsys} is DD under a storage function $\storfunDif$ of the form \cref{7_an_eq:diffStorage} with $M$ satisfying \cref{5_as:mbar} w.r.t. a supply function $\supfunDif$ that satisfies
	\begin{equation}\label{8_eq:supplystabilitydiff}
		\supfunDif(0,\dgp)\leq 0,
	\end{equation}	
	for all $\dgp\in\reals^\gpSize$, then, the nonlinear system is IS. If the supply function satisfies \cref{8_eq:supplystabilitydiff}, but with strict inequality when $\st(t)\neq\sto(t)$, then the nonlinear system is IAS.
\end{theorem}
\begin{proof}
	See \cref{4_pf:difstab}.
\end{proof}

Furthermore, we can also connect this result directly to the DD condition:
\begin{corollary}[Induced incremental stability and invariance]\label{cor:incrinvar}
	For the nonlinear system given by \cref{8_eq:nonlinsys}, let \cref{7_an_eq:diffDissipFull} holds on the convex set $\stSet\times\gdSet \subseteq \reals^{\stSize\cdot\gdSize}$ w.r.t. a supply function $\supfunDif$ that satisfies \cref{8_eq:supplystabilitydiff}, i.e., the system is DD and D(A)S on $\stSet\times\gdSet$.
	 Then, for any input $\gd\in(\gdSet)^\nninteg$, there exists a $\gamma\geq 0$, such that the system is I(A)S and invariant in the tube $\mb{X}_{{\st,\gamma}}(t) \subseteq\stSet,\forall t\in\nninteg$ given by \cref{eq:incrinvarset} around the trajectory $(\st,\gd)$ with $\st\in\proj_\mr{\st}\Bw(\gd)$ and $\st(t)\in\stSet,\forall t\in\nninteg$.
\end{corollary}
The proof simply follows the fact that the system is I(A)S by \cref{4_lem:difstab}, which by \cref{7_an_lem:diffstabcondmi} implies I(A)S under the the (incremental) Lyapunov function $\lyapfunIncr(\st,\sto)=\storfunIncr(\st,\sto)$ with $\storfunIncr$ given by\footnote{This will be derived in the proof of \cref{7_an_thm:ind-incr-dissip}.} \cref{5_eq:jibstorage}. By \cref{thm:incrinvar}, this then implies incremental invariance for any $\gd\in(\gdSet)^\nninteg$ around $\st\in\proj_\mr{\st}\Bw(\gd)$ with $\st(t)\in\stSet$. The implications are very similar to the velocity case, meaning that verifying \cref{7_an_eq:diffDissipFull}  on a bounded set $\stSet\times\gdSet \subseteq \reals^{\stSize\cdot\gdSize}$ only allows one to conclude I(A)S w.r.t an appropriate set of initial conditions $\mb{X}_{{\st,\gamma}}(0)$ in $\stSet$, ensuring invariance (convergence) of all solutions along the signal $(\st,\gd)$. Increasing the sets $\stSet\times\gdSet \subseteq \reals^{\stSize\cdot\gdSize}$, allows one to conclude I(A)S on larger regions of the state space. 

\subsection{Induced incremental dissipativity}
Next, we will present how we can use $\qsr$ differential dissipativity in order to analyze $\qsr$ incremental dissipativity of DT nonlinear systems under mild restrictions of the supply function. 
\begin{theorem}[Induced incremental dissipativity]\label{7_an_thm:ind-incr-dissip}
	If the nonlinear system given by \cref{8_eq:nonlinsys} is $\qsr$-DD %
	with $\supR\preceq0$ under a storage function $\storfunDif$ of the form \cref{7_an_eq:diffStorage}, then the system is $\qsr$-ID for the same tuple $\qsr$.
\end{theorem}
\begin{proof}
	See \cref{7_an_pf:ind-incr-dissip}.
\end{proof}

We can then combine the results of \cref{7_an_thm:ind-incr-dissip,7_an_thm:diffdissip,cor:incrinvar} to arrive at the following \lcnamecref{7_an_thm:incrdissip}:
\begin{corollary}[Incremental dissipativity condition]\label{7_an_thm:incrdissip} 
If for a given $\qsr$ with $\supR\preceq0$ \cref{7_an_eq:diffDissipFull} holds for all $(\stb,\gdb)\in{\stSet}\times{\gdSet}$ and $\stdot\in\stdotSet$ with $\storquad$ satisfying \cref{5_as:mbar}, then 
the nonlinear system given by \cref{8_eq:nonlinsys} is $\qsr$-ID on $\stSet\times\gdSet \subseteq \reals^{\stSize \cdot \gdSize}$ in the following sense. For any inputs $\gd,\gdo\in(\gdSet)^\nninteg$ and the corresponding invariant tubes $\mb{X}_{{\st,\gamma}}(t)$, $\mb{X}_{{\sto,\gamma}}(t)$ in terms of \cref{cor:incrinvar}, all trajectories  $(\st,\gd,\gp),(\sto,\gdo,\gpo)\in\B$ with $\st(0)\in \mb{X}_{{\st,\gamma}}(0)$ and $\sto(0)\in \mb{X}_{{\sto,\gamma}}(0)$ satisfy \eqref{7_an__eq:IDIE} with the  $\qsr$ quadratic supply function.
\end{corollary}
With these results, we have a powerful tool, through the matrix inequality condition in \cref{7_an_thm:diffdissip}, to check $\qsr$-ID of DT systems. In \cite{Koelewijn2021a}, specifically Theorems 10 and 12, it has been shown how in DT, an infinite dimensional set of LMIs can be formulated in order to analyze the incremental \dltwo-gain and incremental passivity of a nonlinear system for a quadratic storage function of the form \cref{7_an_eq:diffStorage} where $\storquad$ is \emph{constant}. Using the results of \cref{7_an_thm:ind-incr-dissip,7_an_thm:incrdissip}, we can now also extend those results to the case with a \emph{state-dependent} matrix $\storquad$:
\begin{corollary}[Incremental \dltwo-gain analysis]\label{7_an_thm:li2gain}
	A nonlinear system given by \cref{8_eq:nonlinsys} has a finite incremental \dltwo-gain of $\perf$ on $\stSet\times\gdSet\subseteq \reals^{\stSize \cdot \gdSize}$, if there exists a matrix function $\storquad$ satisfying \cref{5_as:mbar} such that for all $(\stb,\gdb)\in{\stSet}\times{\gdSet}$ and $\stbdot\in\stdotSet$
	\begin{equation}\label{7_an_eq:matli2}
		\begin{bmatrix}
			\storquad(\stb+\stbdot)\!	& \!\!\difA(\stb,\gdb)\storquad(\stb)\!\!	& \!\!\difB(\stb,\gdb)\!\! 	& 0\\
			\star 		& \storquad(\stb)					& 0					& \!\storquad(\stb)\difC(\stb,\gdb)\!^\top\!\\
			\star		& \star						& \perf I 			& \difD(\stb,\gdb)^\top\\
			\star		& \star						& \star				& \perf I
		\end{bmatrix}\!\!\succeq \!0.
	\end{equation} 
\end{corollary}

\begin{corollary}[DT incremental passivity analysis]\label{7_an_thm:incrpass}
	A nonlinear system given by \cref{8_eq:nonlinsys} is incrementally passive on $\stSet\times\gdSet \subseteq \reals^{\stSize \cdot \gdSize}$, if there exists a matrix function $\storquad$ satisfying \cref{5_as:mbar}, such that for all $(\stb,\gdb)\in{\stSet}\times{\gdSet}$ and $\stbdot\in\stdotSet$
		\begin{equation}\label{7_an_eq:incrpassmatfull}
		\begin{bmatrix}
			\storquad(\stb+\stbdot) 	& \difA(\stb,\gdb)\storquad(\stb) & \difB(\stb,\gdb) \\
			\star 		& \storquad(\stb) 				& \storquad(\stb) \difC(\stb,\gdb)^\top \\
			\star		& \star						& \difD(\stb,\gdb)+(\star)
		\end{bmatrix}\succeq 0.
	\end{equation}
\end{corollary}

The results of \cref{7_an_thm:li2gain,7_an_thm:incrpass} give us conditions for a bounded incremental \dltwo-gain and incremental passivity of nonlinear systems represented by \cref{8_eq:nonlinsys}. These conditions are both given in terms of feasibility checks of an infinite dimensional set of LMIs. As mentioned, in \cref{7_an_sec:lpv}, we will discuss how we can formulate computable tests to check these conditions.

\section{The relation between velocity and differential dissipativity}\label{sec:veloanddiff}
In \cref{8_sec:veloanalysis}, we have shown how (dissipativity) properties of the velocity form \cref{8_eq:veloform} imply universal shifted properties of the primal form \cref{8_eq:nonlinsys} of the nonlinear system. Similarly, in \cref{sec:diffanalysis}, we have shown how dissipativity properties of the differential form \cref{7_an_eq:sys_diff} imply incremental properties of the primal form. In both these cases, we use an alternative form of the system, the velocity and differential form, to imply equilibrium-free properties of the corresponding primal form of the nonlinear system. 

Based on their definitions, it is clear that incremental properties imply universal shifted properties of the system, as incremental properties are properties between all the trajectories, while universal shifted properties are only between trajectories and equilibrium points. Similarly, we also have that $\qsr$-DD implies $\qsr$-VD. Namely, when we compare \cref{7_an_eq:diffDissipFull} in \cref{7_an_thm:diffdissip} to \cref{8_eq:veloMI} in \cref{8_thm:veloqsrMI}, it is evident that these conditions become equivalent in case $\storquad$ in \cref{7_an_eq:diffDissipFull} is a constant matrix. This means that with condition \cref{8_eq:veloMI} for $\qsr$-VD in \cref{8_thm:veloqsrMI}, we actually imply the stronger notion of $\qsr$-DD. This is because in \cref{8_thm:veloqsrMI}, we do not explicitly exploit the fact that $\dtst$, $\dtgd$, and $\dtgp$ represent actual time differences in the state, input, and output, respectively. This leads to conservativeness of the actual result, making it in this case equivalent to checking $\qsr$-DD of the system.

This does not mean that $\qsr$-VD in this form does not have any use. Namely, similarly in the CT case in \cite{Koelewijn2023}, it has been shown how properties of velocity form can still be exploited for controller synthesis in order to achieve closed-loop US stability and performance. Therefore, the velocity based analysis results presented in this paper can serve as a similar stepping stone for developing  controller synthesis algorithms for DT nonlinear systems with such stability and performance guarantees.

\section{Convex Equilibrium-Free Analysis}\label{8_sec:lpvstuffs}\label{7_an_sec:lpv}
In the previous sections, we have shown how the DT velocity and differential forms can be used to imply universal shifted and incremental stability and performance properties of the nonlinear system, respectively. In these sections, we have also discussed how through the (infinite dimensional set of) LMIs in \cref{8_eq:veloMI,7_an_eq:diffDissipFull}, $\qsr$-VD and $\qsr$-DD can be analyzed. In the analysis of LPV systems, similar problems are encountered due to the variation of the scheduling-variable for which various tools have been developed  to make these problems computationally tractable. Hence, as for the CT case in \cite{Verhoek2020, Koelewijn2023, Koelewijn2021a}, we can use the analysis results of the LPV framework to turn the proposed checks for universal shifted and incremental stability and performance analysis to convex finite dimensional optimization problems which can be efficiently solved as SDPs.

As discussed in \cref{sec:veloanddiff}, for a constant matrix $M$, the matrix inequalities in \cref{8_eq:veloMI,7_an_eq:diffDissipFull} are equivalent. Therefore, we will only discuss the $\qsr$-DD case \cref{7_an_eq:diffDissipFull}, as the $\qsr$-VD case trivially follows from it.
Furthermore, we discussed in \cref{sec:diffanalysis} that the $\qsr$-DD condition can be interpreted as analyzing classical $\qsr$ dissipativity of the differential form \cref{7_an_eq:sys_diff}. Therefore, to cast the $\qsr$-DD analysis problem to an LPV analysis problem, we embed the differential form of the nonlinear system in an LPV representation, which we call a \emph{Differential Parameter-Varying} (DPV) embedding of the nonlinear system \cref{8_eq:nonlinsys}:
\begin{definition}[DPV embedding]\label{7_an_def:lpvemb}
	Given a nonlinear system in the form of \cref{8_eq:nonlinsys} with differential form given by \cref{7_an_eq:sys_diff}. Then, the LPV representation given by
	\begin{subequations}\label{7_an_eq:lpvdiff}
		\begin{align}
		\dst(t+1)&=\lpvA(\sch(t))\dst(t) + \lpvB(\sch(t))\dgd(t),\\ 
		\dgp(t) &= \lpvC(\sch(t))\dst(t) + \lpvD(\sch(t))\dgd(t),
	\end{align}
\end{subequations}
	where $\sch(t)\in\schSet\subset\reals^\schSize$ is the scheduling-variable, is a DPV embedding of \cref{8_eq:nonlinsys} on the compact convex region $\stSetLPV\times\gdSetLPV\,\subseteq\reals^\stSize\!\times\reals^\gdSize\!$, if there exists a function $\schMap:\stSetLPV\times\gdSetLPV\to\schSet$, called the scheduling-map,  such that under a given choice of function class for $\lpvA,\,\dots,\,\lpvD$, e.g. affine, polynomial, etc., $\lpvA(\schMap(\stb,\gdb))=\difA(\stb,\gdb)$, $\dots$, $\lpvD(\schMap(\stb,\gdb))=\difD(\stb,\gdb)$ for all $(\stb,\gdb)\in\stSetLPV\times\gdSetLPV$ and $\schMap(\stSetLPV,\gdSetLPV)\subseteq \schSet$.%
\end{definition}
	
Let us denote by $\schdot(t)=\sch(t+1)-\sch(t)\in\schdotSet$. We assume that the set $\schdotSet$ is considered such that it includes $(\st(t+1)-\st(t))\in\stdotSet$.
Through the DPV embedding, we can then cast the $\qsr$-DD (and $\qsr$-VD) check as an LPV analysis problem:
\begin{theorem}[LPV based $\qsr$-DD condition]\label{7_an_thm:diffdissipLPV}
	A nonlinear system given by \cref{8_eq:nonlinsys} with a corresponding DPV embedding given by \cref{7_an_eq:lpvdiff} on $\stSetLPV\times\gdSetLPV\,\subseteq\reals^\stSize\!\times\reals^\gdSize\!$ is $\qsr$-DD on $\stSetLPV\times\gdSetLPV$, if there exists a matrix function $\storquad:\schSet\to\sym^{\stSize}$ satisfying \cref{5_as:mbar}, such that	\begin{multline}\label{7_an_eq:diffDissipFullLPV}
	(\star)^\top \begin{bmatrix}
		-\storquad(\sch) & {0}\\{0} & \storquad(\sch+\schdot)
	\end{bmatrix}\begin{bmatrix}
		I & {0}\\\lpvA(\sch) & \lpvB(\sch)
	\end{bmatrix}-\\(\star)^\top \qsrMat \begin{bmatrix}
		{0} & I\\\lpvC(\sch) & \lpvD(\sch)
	\end{bmatrix}\preceq 0,
	\end{multline}
	 for all $\sch\in\schSet$ and $\schdot\in\schdotSet$.
\end{theorem}
\begin{proof}
See \cref{7_an_pf:diffdissipLPV}.
\end{proof}

Note that the proposed analysis condition \cref{7_an_eq:diffDissipFullLPV} can be seen as a classical $\qsr$ dissipativity analysis condition of an LPV representation. Therefore, all the techniques from the LPV framework can be used to reduce the infinite dimensional set of LMIs to a finite set. The most common techniques for this are polytopic techniques \cite{DeCaigny2012,Cox2018a}, grid-based techniques \cite{Wu1995,Apkarian1998}, and multiplier-based techniques \cite{Scorletti1998,Scherer2001}, see also \cite{Hoffmann2015} for an overview. For the polytopic and multiplier-based techniques, $\lpvA,\,\dots,\,\lpvD$ are needed to be restricted to an affine or rational function in the embedding \cref{7_an_eq:lpvdiff}, respectively. With the results of \cref{7_an_thm:diffdissipLPV}, combined with \cref{7_an_thm:ind-incr-dissip}, we have a convex analysis condition in order to analyze $\qsr$-ID of DT nonlinear systems. Similarly, connecting to \cref{8_thm:veloqsrMI}, we then also have convex analysis tools for universal shifted stability, through \cref{8_thm:velotoshiftstab}, and performance, through \cref{8_thm:veloshiftperf}. Note that in these cases, it is important that we also induce an invariant set in $\stSetLPV\times\gdSetLPV$ through the result of \cref{cor:shiftinvar,cor:incrinvar}. As this invariant set describes the region where the implied stability and performance conditions hold, we can also maximize its volume. %
For example, for a constant matrix $M$, the invariant set will correspond to an ellipsoidal region. There exist various result on casting the maximization of the volume of an ellipsoid as a convex problem, e.g., see \cite{Boyd2004}. However, this is outside the scope of this paper.

\begin{figure}
\centering
\includegraphics[scale=0.9]{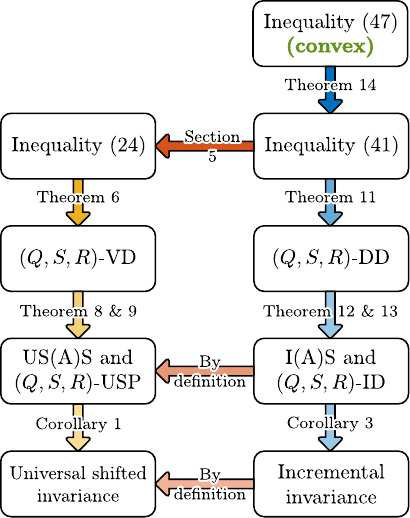}
\caption{Overview of the results and their connections.}
\vspace{-.5em}
\label{fig:overview}
\end{figure}

Also note that although the same tools from the LPV framework can be used for checking $\qsr$-DD, $\qsr$-VD and classical $\qsr$ dissipativity of nonlinear systems, we would like to emphasize that the underlying dissipativity and stability concepts and the matrix functions on which these sets are applied are very different. Namely, the $\qsr$-DD and $\qsr$-VD concepts connect to the equilibrium-free concepts of incremental and universal shifted stability and performance. This means that these analysis results are not dependent on a particular trajectory or equilibrium point, respectively. On the other hand, the standard LPV analysis results applied on a direct LPV embedding of a nonlinear system use classical dissipativity and can only provide performance and stability analysis with respect to single equilibrium point, often the origin of the state-space representation of the nonlinear system, which make them \emph{equilibrium-dependent}.

An overview of all the results and their connections presented in this paper is given in \cref{fig:overview}.
 
\section{Example}\label{7_an_sec:example}
In this section, we apply the results of the previous sections in order to analyze equilibrium-free stability and performance of a discrete-time nonlinear system.

We consider the following CT state-space representation of an actuated Duffing oscillator:
\begin{subequations}\label{eq:ctsysexample}
	\begin{align}
		\begin{bmatrix}
			\dot{x}_1 \\\dot{x}_2
		\end{bmatrix} &= \begin{bmatrix}
			x_2\\
			-8 x_1 - 10 x_1^3 - 4 x_2 + \gd
		\end{bmatrix};\\
		\gp &= x_1,
	\end{align}
\end{subequations}
where $x_1$ represents the position of the spring, $x_2$ its velocity, and $\gd$ is an input force. We discretize this model using a fourth order Runge-Kutta (RK4) method with a sampling time of 0.01s, where we assume the input $\gd$ to be constant in between samples. The resulting model is of the form \cref{7_an_eq:nl} and is not given due to its complexity. For our analysis, we consider the operating region $x_1(t)\in[-1\;1]$, $x_2(t)\in[-1\;1]$, $\gd(t)\in[-1\;1]$ for all $t\in\nninteg$, i.e., we perform dissipativity analysis on the set $\stSet\times\gdSet$ with $\stSet = [-1\;1]\times[-1\;1]$ and $\gdSet = [-1\;1]$.  Based on these sets, we compute the corresponding set $\stdotSet$, which is given by: $\stdotSet = [-0.011\;0.011]\times[-0.23\;0.23]$ such that $\st(t+1)-\st(t)\in\stdotSet$.  Consequently, we construct a DPV embedding of the nonlinear model on $\stSetLPV\times\gdSetLPV = [-1\;1]\times[-1\;1]\times[-1\;1]$. If we do a direct DPV embedding, we obtain as scheduling-map $\schMap(x,w) = \col(x_1, x_2, w)$, therefore $\sch(t) =\col(x_1(t), x_2(t), w(t))\in\schSet$ where $\schSet = [-1\;1]\times[-1\;1]\times[-1\;1]$. Moreover, we consider $\schdotSet = \stdotSet\times\reals$, such that $\schdot(t) = {\sch(t+1)}-\sch(t)\in\schdotSet$.

Under these considerations, we minimize $\perf$ in \cref{7_an_thm:diffdissipLPV} with $\qsr = (\perf^2,0,-I)$, corresponding to an incremental and universal shifted \dltwo-gain bound of $\perf$. We use a grid-based LPV approach and consider our quadratic storage matrix $M$ to be of the form:\footnote{Note that $\storquad$ only depends on the scheduling-variable $p_1$ which corresponds to the state $x_1$, consistent with \cref{7_an_eq:diffDissipFull}.}
\begin{equation}
	M(\sch) = M_0 + M_1 \sch_1^2
\end{equation}
where $M_i\in\reals^{2\times 2}$ for $i=1,\dots, 2$. Our problem then corresponds to grid-based \dltwo-gain analysis of an LPV representation, which has been implemented in the LPVcore Toolbox \cite{DenBoef2021}. Using the LPVcore Toolbox, the resulting $\perf$ that we obtain is 0.13, which is then our upperbound for the incremental and universal shifted \dltwo-gain bound of the discretized version of \cref{eq:ctsysexample} on the region $[-1\;1]\times[-1\;1]\times[-1\;1]$.

Computing $\perf$ for a constant quadratic storage matrix $M$ only results in a upperbound for the incremental and universal shifted \dltwo-gain of 0.42. Therefore, this shows that for this example, the approach using a constant matrix $\storquad$ in the storage function presented in \cite{Koelewijn2021a} is much more conservative than the approach using a state-dependent storage function presented in this paper for analyzing the incremental and universal shifted \dltwo-gain of DT systems. 

We also simulate the system for two different inputs to additionally verify the equilibrium-free properties:
\begin{align}
	\gd(t) = 0.7 e^{-t} \sin(2t) + 0.3 \sin(0.2 t),\\
	\tilde\gd(t) = 0.3 e^{-t} \cos(t) + 0.3 \sin(0.2 t),
\end{align}
and initial conditions $\st_{0} = \col(-0.08, 0.22)$ and \linebreak$\tilde\st_{0} =\col(-0.50, -0.20)$, respectively. 
Note that these two input trajectories converge as $t\rightarrow \infty$. Therefore, as the system is incrementally stable, the state and output trajectories will also converge, as is visible in \cref{fig:statetraj}. The cumulative incremental supply plus the initial storage $ \sum_{\tau=0}^{t}\supfunIncr \big(\gd(\tau), \gdo(\tau),\gp(\tau), \gpo(\tau)\big) + {\storfunIncr}\big(\st(0), \sto(0)\big)$ and the incremental storage ${\storfunIncr}\big(\st(t+1),\sto(t+1)\big)$ are also plotted in \cref{fig:stor}. The incremental storage ${\storfunIncr}$ is obtained by solving \cref{5_eq:geodesic} and using it in \cref{5_eq:jibstorage}. To simplify this computation, the integrals are approximated by summations and the smooth path $\stbmin_{(\st,\sto)}$ is approximated as a piecewise linear function. We then solve problem \cref{5_eq:geodesic} using \texttt{fmincon} in MATLAB. From \cref{fig:stor}, it can be seen that the incremental storage is always smaller than or equal to the cumulative incremental supply plus the initial storage. Therefore, this also verifies that the system satisfies the ID inequality, see \cref{7_an__eq:IDIE}, for these trajectories.

\begin{figure}
	\centering
	\includegraphics[width=\columnwidth]{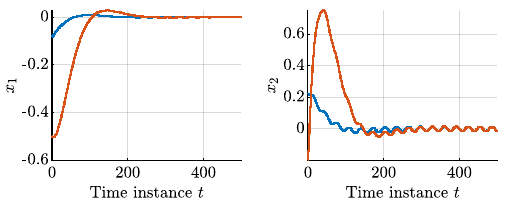}
	\caption{State trajectories for input $\gd$ with initial condition $\st_{0}$ \mbox{(\legendline{mblue})} and for input $\gdo$ with initial condition $\tilde\st_{0}$ (\legendline{morange}).}
	\label{fig:statetraj}
\end{figure} 
\begin{figure}
	\centering
	\includegraphics[width=\columnwidth]{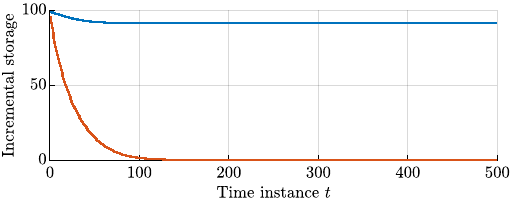}
	\caption{The cumulative incremental supply plus the initial storage $ \sum_{\tau=0}^{t}\supfunIncr \big(\gd(\tau), \gdo(\tau),\gp(\tau), \gpo(\tau)\big) + {\storfunIncr}\big(\st(0), \sto(0)\big)$ (\legendline{mblue}) and the incremental storage ${\storfunIncr}\big(\st(t+1),\sto(t+1)\big)$ (\legendline{morange}) for the trajectories generated by the inputs $\gd$ and $\gdo$.}
	\label{fig:stor}
\end{figure} 

\section{Conclusions}\label{8_sec:concl}

In this paper, we have developed convex conditions for equilibrium-free analysis of discrete-time nonlinear systems. We have shown how time-difference dynamics can be used in order to analyze universal shifted stability and performance of discrete-time nonlinear systems. Similarly, we have shown how dissipativity of the differential form can be used in order to analyze incremental dissipativity using a state-dependent storage function. Finally, we have shown how both these analysis results can be cast as an analysis problem of an LPV representation. These results give us convex tools for equilibrium-free stability and performance analysis of discrete-time nonlinear systems. For future research, we aim to use these results in order to develop equilibrium-free controller synthesis methods for discrete-time nonlinear systems.

\section*{Appendices}
\appendices
\section{Proofs}

\proofsection{8_thm:shiftlyapstab}\label{8_pf:shiftlyapstab}
Consider the function $\lyapfun:\st\mapsto\lyapfunShift(\st,\gdEq)$, which for every $(\stEq,\gdEq)\in\proj_\mr{\stEq,\gdEq}\eqSet$ satisfies the conditions for a Lyapunov function for the equilibrium point $\stEq$. Namely, given a $(\stEq,\gdEq)\in\proj_\mr{\stEq,\gdEq}\eqSet$, we have that $\lyapfun = (\st\mapsto\lyapfunShift(\st,\gdEq))\in\posClass{\stEq}$. Therefore, for a given $(\stEq,\gdEq)\in\proj_\mr{\stEq,\gdEq}\eqSet$, by \cref{8_eq:shiftedstability}, it holds that
	\begin{equation}\label{4_eq:lyaptoshift}
		\lyapfun(\st(t+1))-\lyapfun(\st(t))\leq 0,
	\end{equation}
	for all $t\in\nninteg$ and $\st\in\proj_\mr{\st}\Bw(\gd\equiv\gdEq)$. Due to the properties of $\lyapfunShift$ and construction of $\lyapfun$, \cref{4_eq:lyaptoshift} then also holds for each $(\stEq,\gdEq)\in\proj_\mr{\stEq,\gdEq}\eqSet$. Consequently, each equilibrium point \linebreak${(\stEq,\gdEq,\gpEq)\in\eqSet}$ is stable, under $\gd\equiv\gdEq$, for the whole state-space by the Lyapunov theory, see e.g. \cite{Khalil2002, Bof2018}. Therefore, by \cref{4_def:shiftedstab}, it is USS. In a similar manner, if \cref{8_eq:shiftedstability} holds, but with a strict inequality except when $\st(t)=\stEq$, this implies that \cref{4_eq:lyaptoshift} holds. %
	This then implies asymptotic stability of each equilibrium point, meaning that the system is USAS.

\proofsection{thm:shiftinvar}\label{pf:USinvariance}
Given $\gdEq\in\gdSetEq$ and a $\gamma >0$, define the set \cref{eq:shiftinvarset}. 
	We have that \cref{8_eq:shiftedstability} holds for every $(\stEq,\gdEq)\in\proj_\mr{\stEq,\gdEq}\eqSet$ and for all $t\in\nninteg$ and $\st\in\proj_\mr{\st}\Bw(\gd\equiv \gdEq)$.
	Therefore, for $x(0)\in\mb{X}_{\gdEq,\gamma}$, it holds that
	\begin{equation}\label{eq:shiftinvarlyap}
		\lyapfunShift(\st(t),\gdEq)\leq\dots\leq \lyapfunShift(\st(1),\gdEq)\leq \lyapfunShift(\st(0),\gdEq)\leq \gamma,
	\end{equation}
	for $t\geq 1$.
	Consequently, for the nonlinear system given by \cref{8_eq:nonlinsys} with initial condition $\st(0)=\stIc$ and input $\gd\equiv\gdEq$, we have by \cref{eq:shiftinvarlyap} that $\st(t)=\sttran(t,0,\stIc,\gd\equiv\gdEq)\in\mb{X}_{\gdEq,\gamma}$ for all $t\in\nninteg$.
	
\proofsection{thm:incrinvar}\label{pf:IncrInvariance}
	The proof follows similarly as the proof for \cref{thm:shiftinvar}. Given a trajectory $(\sto,\gd)\in\proj_\mr{\st,\gd}\in\Bw(\gd)$ and a $\gamma>0$ define the time-varying set \cref{eq:incrinvarset}. We have that \cref{7_an_eq:incrlyap} holds for all $t\in\nninteg$ and $\st,\sto\in\proj_{\mr{\st}}\Bw(\gd)$ under all measurable and bounded  $\gd\in(\reals^\gpSize)^{\nninteg}$.	Therefore, for $\st(0)\in\mb{X}_{\sto,\gamma}(0)$, it holds that
	\begin{equation}\label{eq:incrinvarlyap}
		\lyapfunIncr(\st(t),\sto(t))\leq\dots\leq \lyapfunIncr(\st(1),\sto(1))\leq \lyapfunIncr(\st(0),\sto(0)) \leq \gamma,
	\end{equation}
	for $t\geq 1$. 
	Consequently, for the nonlinear system given by \cref{8_eq:nonlinsys} with initial condition $\st(0)=\stIc$ and input $\gd$, we have by \eqref{eq:incrinvarlyap} that $\st(t)=\sttran(t,0,\stIc,\gd)\in\mb{X}_{\sto,\gamma}(t)$ for all $t\in\nninteg$.

\proofsection{8_thm:veloqsrMI}\label{8_pf:veloqsrMI}
If \cref{8_eq:veloMI} holds for all $(\st,\gd)\in\stSet\times\gdSet$, we have by pre- and post multiplication of \cref{8_eq:veloMI} with $\col(\dtst,\dtgd)^\top$ and $\col(\dtst,\dtgd)$, respectively, that
\begin{multline}\label{eq:8_pf_veloqsrMI_1}
	(\star)^\top \storquad (\velA(\st,\gd)\dtst+\velB(\st,\gd)\dtgd)-\dtst^\top \storquad\dtst - \\ \dtgd^\top \supQ \dtgd - 2 \dtgd^\top \supS \big(\velC(\st,\gd)\dtst + \velD(\st,\gd)\dtgd\big)-\\
(\star)^\top  \supR \big(\velC(\st,\gd)\dtst + \velD(\st,\gd)\dtgd\big) \leq 0,
\end{multline}
for all $\dtst\in\reals^\stSize$, $\dtgd\in\reals^\gdSize$ and $(\st,\gd)\in\stSet\times\gdSet$. As $\stSet$ and $\gdSet$ are assumed to be convex, we can represent any $\stb\in\stSet$ and $\gdb\in\gdSet$ by a $\var\in[0,1]$, $\stpl,\st\in\stSet$, and $\gdpl,\gd\in\gdSet$, such that $\stb(\var) = \st+\var(\stpl-\st)$ and $\gdb(\lambda)=\gd+\var(\gdpl-\gd)$. Consequently, if \cref{eq:8_pf_veloqsrMI_1} holds, it also holds that
\begin{multline}
	(\star)^\top \storquad (\velA(\stb(\var),\gdb(\var))\dtst+\velB(\stb(\var),\gdb(\var))\dtgd)-\\\dtst^\top \storquad\dtst - \dtgd^\top \supQ \dtgd - \\2 \dtgd^\top \supS \big(\velC(\stb(\var),\gdb(\var))\dtst + \velD(\stb(\var),\gdb(\var))\dtgd\big)-\\
(\star)^\top  \supR \big(\velC(\stb(\var),\gdb(\var))\dtst + \velD(\stb(\var),\gdb(\var))\dtgd\big) \leq 0,
\end{multline}
for any $\lambda\in[0,1]$, $\stpl,\st\in\stSet$, $\gdpl,\gd\in\gdSet$, $\dtst\in\reals^\stSize$ and $\dtgd\in\reals^\gdSize$. Hence, we also have by integration over $\var$ that
\begin{multline}\label{8_eq:mipf0}
	\int_0^1 (\star)^\top \storquad (\velA(\stb(\var),\gdb(\var))\dtst+\velB(\stb(\var),\gdb(\var))\dtgd)-\\\dtst^\top \storquad\dtst - \dtgd^\top \supQ \dtgd -\\ 2 \dtgd^\top \supS \big(\velC(\stb(\var),\gdb(\var))\dtst + \velD(\stb(\var),\gdb(\var))\dtgd\big)-\\
(\star)^\top  \supR \big(\velC(\stb(\var),\gdb(\var))\dtst + \velD(\stb(\var),\gdb(\var))\dtgd\big)\,d\var \leq 0,
\end{multline}
for any $\stpl,\st\in\stSet$, $\gdpl,\gd\in\gdSet$, $\dtst\in\reals^\stSize$ and $\dtgd\in\reals^\gdSize$. By \cite[Lemma 16]{Koelewijn2021a}, as $\storquad \succ 0$, we have that
\begin{multline}\label{8_eq:mipf1}
\!\!\! (\star)^\top\! \storquad \!\left(\int_0^1 \!\velA(\stb(\var),\gdb(\var))\dtst\!+\!\velB(\stb(\var),\gdb(\var))\dtgd \,d\var\!\right) \leq\\
	\int_0^1 (\star)^\top \storquad (\velA(\stb(\var),\gdb(\var))\dtst+\velB(\stb(\var),\gdb(\var))\dtgd) \,d\var,
\end{multline}
and similarly, as $\supR \preceq 0$, we have that 
\begin{multline}\label{8_eq:mipf1a}
\!\!\!\!\!\!(\star)^\top  \!(\text{-}\supR) \!\left(\int_0^1\! \velC(\stb(\var),\gdb(\var))\dtst \!+\! \velD(\stb(\var),\gdb(\var))\dtgd\,d\var\!\right)\!\leq \\
\int_0^1	(\star)^\top  (\text{-}\supR) \big(\velC(\stb(\var),\gdb(\var))\dtst + \velD(\stb(\var),\gdb(\var))\dtgd\big)\,d\var.\!\!
\end{multline}
Note that $\velA=\Partial{\stMap}{\st}$, $\velB=\Partial{\stMap}{\gd}$, $\velC=\Partial{\opMap}{\st}$, $\velD=\Partial{\opMap}{\gd}$. Hence, using the definition of $\vintA,\dots,\vintD$ in \cref{8_eq:velointmat}, we have
\begin{multline}\label{8_eq:mipf1b}
	\int_0^1 \velA(\stb(\var),\gdb(\var))\dtst+\velB(\stb(\var),\gdb(\var))\dtgd \,d\var = \\
	\vintA(\stpl,\st,\gdpl,\gd)\dtst + \vintB(\stpl,\st,\gdpl,\gd)\dtgd,
\end{multline}
\begin{multline}\label{8_eq:mipf2}
	\int_0^1 \velC(\stb(\var),\gdb(\var))\dtst+\velD(\stb(\var),\gdb(\var))\dtgd \,d\var = \\
	\vintC(\stpl,\st,\gdpl,\gd)\dtst + \vintD(\stpl,\st,\gdpl,\gd)\dtgd.
\end{multline}
Combining \cref{8_eq:mipf1,8_eq:mipf1a,8_eq:mipf1b,8_eq:mipf2} with \cref{8_eq:mipf0}, we obtain that
\begin{multline}\label{8_eq:velo_dissip_valuepffff}
		(\star)^\top \storquad \left(\vintA(\stpl,\st,\gdpl,\gd)\dtst +\vintB(\stpl,\st,\gdpl,\gd)\dtgd\right)-\\\ \dtst^\top \storquad\dtst- 
		\dtgd^\top \supQ \dtgd -\\ 2 \dtgd^\top\supS\big(\vintC(\stpl,\st,\gdpl,\gd)\dtst + \vintD(\stpl,\st,\gdpl,\gd)\dtgd\big)-\\
		(\star)^\top  \supR\big(\vintC(\stpl,\st,\gdpl,\gd)\dtst + \vintD(\stpl,\st,\gdpl,\gd)\dtgd\big) \leq 0,
\end{multline}
for any $\stpl,\st\in\stSet$, $\gdpl,\gd\in\gdSet$, $\dtst\in\reals^\stSize$ and $\dtgd\in\reals^\gdSize$. Substituting $\stpl = \st(t+1)$, $\st = \st(t)$, $\dtst = \st(t+1)-\st(t)$, $\gd = \gd(t)$, $\dtgd = \gd(t+1)-\gd(t)$ in \cref{8_eq:velo_dissip_valuepffff} and summing over time from $t_0$ to $t_1$ where $t_0 \leq t_1$, we obtain \cref{8_eq:velodissip} where $\storfunVelo$ is given by \cref{8_eq:velostorquad} and $\supfunVelo$ is given by \cref{8_eq:velosupply}. 
	
\proofsection{8_thm:velotoshiftstab}\label{8_pf:velotoshiftstab}
For each equilibrium point $(\stEq,\gdEq,\gpEq)\in\eqSet$, consider
\begin{equation}
	\lyapfunShift(\st(t),\gdEq):=\lyapfunVelo(\stMap(\st(t),\gdEq)-\st(t))=\lyapfunVelo(\dtst(t)). 
\end{equation}
For each $(\stEq,\gdEq,\gpEq)\in\eqSet$, this choice implies that $\storfunShift(\cdot,\gdEq)\in\posClass{\stEq}$, as $\storfunVelo\in\posClass{0}$. Note that this requires uniqueness of the equilibrium points (see \cref{4_assum:uniqueEq}), as otherwise there exists multiple $\stEq$ for which $\lyapfunShift(\stEq,\gdEq)=0$. By this choice of $\lyapfunShift$, we have that for each $(\stEq,\gdEq,\gpEq)\in\eqSet$,
\begin{multline}\label{8_eq:shiftvelostab}
	\lyapfunShift(\st(t+1),\gdEq)-\lyapfunShift(\st(t),\gdEq) =\\ \lyapfunVelo(\dtst(t+1))-\lyapfunVelo(\dtst(t))\leq 0,
\end{multline}
for all $t\in\nninteg$ and $\dtst\in\proj_\mr{\dtst}\Bvw(\gd \equiv \gdEq)$ and correspondingly for all $\st\in\proj_\mr{\st}\Bw(\gd\equiv\gdEq)$. This implies that \cref{8_eq:shiftedstability} holds for all $\st\in\proj_\mr{\st}\Bw(\gd\equiv\gdEq)$ and for all equilibrium points $(\stEq,\gdEq)\in\proj_\mr{\stEq,\gdEq}\eqSet$. Hence, by \cref{8_thm:shiftlyapstab}, \cref{8_eq:nonlinsys} is USS. USAS follows similarly by changing \cref{8_eq:shiftvelostab} to a strict inequality.

\proofsection{4_lem:velostab}\label{4_pf:velostab}
Let the system given by \cref{8_eq:nonlinsys} be velocity dissipative w.r.t. a supply function $\supfunVelo$. For this supply function, \cref{8_eq:supplystability} holds for all $\gpdot\in\reals^\gpSize$. Therefore, it holds that
\begin{equation}\label{4_eq:storfunshiftproof}
\storfunVelo(\dtst(t_1+1))-\storfunVelo(\dtst(t_0))\leq \sum_{t=t_0}^{t_1} \supfunVelo(\dtgd(t),\dtgp(t)) \leq 0
\end{equation}
for all $t_0,t_1\in\nninteg$ with $t_0\leq t_1$ and $(\dtst,\dtgd,\dtgp)\in\Bv$. This gives that
 \begin{equation}\label{4_eq:storfunshiftproof2}
	\storfunVelo(\dtst(t+1))-\storfunVelo(\dtst(t))\leq 0,
 \end{equation}
	for all $t\in\nninteg$ and $\dtst\in\proj_\mr{\dtst}\Bvset{\gdSetEq}$. Moreover, the storage function $\storfunVelo$ satisfies the conditions for the function $\lyapfunVelo$ in \cref{8_thm:velotoshiftstab}. Hence, by \cref{8_thm:velotoshiftstab}, \cref{4_eq:storfunshiftproof2} implies that the  system is USS.

In case of USAS, the supply function satisfies \cref{8_eq:supplystability}, but with strict inequality for all $\gpdot\in\reals^\gpSize$, except when $\dtst=0$. Therefore, \cref{4_eq:storfunshiftproof2} holds, but with strict inequality except when $\dtst(t)=0$, which by \cref{8_thm:velotoshiftstab} implies USAS.

\proofsection{8_thm:veloshiftperf}\label{8_pf:veloshiftperf}
If the nonlinear system given by \cref{8_eq:nonlinsysState} is velocity dissipative w.r.t. the supply function $\supfunVelo(\dtgd,\dtgp) = \dtgd ^\top \supQ \dtgd+\dtgp^\top \supR \dtgp$, then there exists a function $\storfunVelo$, such that for all $t_0,t_1\in\nninteg$ with $t_0\leq t_1$
\begin{multline}
	\storfunVelo(\dtst(t_1+1))-\storfunVelo(\dtst(t_0))\leq\\ \sum_{t=t_0}^{t_1} \dtgd(t)\!^\top \!\supQ \dtgd(t)+\dtgp(t)\!^\top\! \supR \dtgp(t),
\end{multline}
for all $(\dtst,\dtgd,\dtgp)\in\Bv$, corresponding to $(\st,\gd,\gp)\in\B$. Note that by consideration of the theorem, $\dtst(0)=0$. Hence, as $\storfunVelo(\dtst(0))=\storfunVelo(0)=0$ and $\storfunVelo(\stdot)\geq 0$, $\forall\, \stdot\in\reals^\stSize$, this implies that
\begin{equation}\label{8_eq:pf:qrvsp1}
	0\leq \sum_{t=0}^{T} \dtgd(t)\!^\top \! \supQ \dtgd(t)+\dtgp(t)\!^\top \! \supR \dtgp(t),
\end{equation}
for all $T\geq0$ and $(\dtst,\dtgd,\dtgp)\in\Bv$. 
Defining $\tilde \supQ:=\frac{1}{\norm{\supQ}}\supQ$ and $\tilde \supR:=\frac{1}{\norm{\supQ}}\supR$, it also holds that 
\begin{equation}\label{8_eq:pf:qrvsp1.1}
	0\leq \sum_{t=0}^{T}\dtgd(t) \!^\top \! \tilde \supQ \dtgd(t)+\dtgp(t)\!^\top \! \tilde \supR \dtgp(t),
\end{equation}
Next, using \cref{8_eq:nonlinsysState,8_eq:nonlinsysStateEqui,8_eq:nonlinsysStateVelo} and as $\dtst(t) = \st(t+1)-\st(t)$, we have that, omitting dependence on time for brevity,
\begin{align}
	\dtgp^\top \tilde \supR \dtgp \!&=\! \dtst^\top \ltiC^\top \tilde \supR\, \ltiC\dtst, \label{8_eq:perfproofeq1}\\
	\!&=\! (\star)^\top\! \tilde \supR\, \ltiC(\stMap(\st)+\ltiB \gd-\st), \notag\\[-4mm]
		\!&=\! (\star)^\top\! \tilde \supR\,\ltiC(\stMap(\st)+\ltiB \gd-\st+\overbrace{\stEq-(\stMap(\stEq)+\ltiB \gdEq)}^{=0}),\notag \\
	\!&=\! (\star)^\top\! \tilde \supR\,\ltiC(\stMap(\st)-\stMap(\stEq)-(\st-\stEq)+\ltiB (\gd -\gdEq)). \notag
\end{align}
Through the Fundamental Theorem of Calculus \cite{Thomas2005}, we have that 
\begin{align}
	\stMap(\st)-\stMap(\stEq)&=\left(\int_0^1\Partial{\stMap}{\st}(\stEq+\var(\st-\stEq))\,d\var\right)(\st-\stEq),  \notag \\
	&= \underbrace{\left(\int_0^1 \velA(\stEq+\var(\st-\stEq))\,d\var\right)}_{\vintA(\st,\stEq)}(\st-\stEq),
\end{align} \vskip -3mm \noindent
hence,
\begin{equation}
	\stMap(\st)-\stMap(\stEq)-(\st-\stEq) = (\vintA(\st,\stEq)-I)(\st-\stEq).
\end{equation}
Combining this with \cref{4_as:CB}, we can write \cref{8_eq:perfproofeq1} as
\begin{equation}
	\dtgp^\top \tilde \supR \dtgp = (\star)\!^\top \!\tilde \supR\,\ltiC(\vintA(\st,\stEq)-I)(\st-\stEq).
\end{equation}
Next, by satisfying \cref{8_as:veloShiftBound2} for $T=\tilde\supR\preceq 0$, we have that, for every $\stEq\in\stSetEq$,
\begin{multline}\label{8_eq:pf:qrvsp2}
	\dtgp^\top \tilde \supR \dtgp= (\star)\!^\top \! \tilde \supR \ltiC(\vintA(\st,\stEq)-I)(\st-\stEq) \leq\\
	\alpha^{-1}(\star)\!^\top \! \tilde \supR \ltiC(\st-\stEq)= \alpha^{-1}(\star)^\top \tilde \supR(\gp-\gpEq).
\end{multline}
Moreover, by \cref{8_as:wExoSys} and using that $\dtgd(t) = \gd(t+1)-\gd(t)$, we have that, for a given $(\stEq,\gdEq,\gpEq)\in\eqSet$,
\begin{align}
	\gd(t+1) &= \exoA(\gd(t)-\gdEq)+\gdEq, \notag \\
	\gd(t+1)-\gd(t)+\gd(t)&=\exoA(\gd(t)-\gdEq)+\gdEq, \notag \\
	\gd(t+1)-\gd(t)&=\exoA(\gd(t)-\gdEq)-(\gd(t)-\gdEq), \notag \\
	\dtgd(t) &= (\exoA-I)(\gd(t)-\gdEq),
\end{align}
and hence,
\begin{multline}\label{8_eq:pf:qrvsp3}
	\dtgd(t)\!^\top \!\tilde \supQ \dtgd(t)=(\star)\!^\top \!\tilde \supQ (\exoA-I)(\gd(t)-\gdEq)\leq\\ \beta^2 (\star)\!^\top \! \tilde \supQ (\gd(t)-\gdEq),
\end{multline}
where $\gd\in\exoBvr_{(\gdEq,\beta)}$ and $0\preceq\tilde \supQ \preceq I$.
Combining \cref{8_eq:pf:qrvsp1.1,8_eq:pf:qrvsp2,8_eq:pf:qrvsp3}, we obtain that, for every $(\stEq,\gdEq,\gpEq)\in\eqSet$,
\begin{equation}
	\sum_{t=0}^{T} \beta^2(\star)\!^\top \! \tilde \supQ (\gd(t)-\gdEq)+\alpha^{-1}(\star)\!^\top \! \tilde \supR (\gp(t)-\gpEq)\geq 0,
\end{equation}
for all $T\geq 0$ and $(\gd,\gp)\in\proj_\mr{\gd,\gp}\B$ with $\gd\in\exoBvr_{(\gdEq,\beta)}$.
Hence, also 
\begin{equation}
	\sum_{t=0}^{T} \beta^2(\star)\!^\top \! \supQ (\gd(t)-\gdEq)+\alpha^{-1}(\star)\!^\top \!  \supR (\gp(t)-\gpEq) \geq 0,
\end{equation}
for all $T\geq 0$ and $(\gd,\gp)\in\proj_\mr{\gd,\gp}\B$ with $\gd\in\exoBvr_{(\gdEq,\beta)}$.

\proofsection{7_an_thm:diffdissip}\label{7_an_pf:diffdissip}
The system given by \cref{7_an_eq:nl} is differentially dissipative w.r.t. a supply function $\supfunDif$ and for a storage function $\storfunDif$, if \cref{7_an_eq:DIE_diff} holds for all $(\stb,\gdb)\in\proj_\mr{x,u}\B$ and for all $t_0,t_1\in\nninteg$ with $t_0 \le t_1$. This condition is equivalent to 
\begin{equation}\label{7_an_eq:proofdiffdissip1}
	{\storfunDif}\big(\stb(t+1),\dst(t+1)\big) - \storfunDif\big(\stb(t),\dst(t)\big) \leq \supfunDif\big(\dgd(t),\dgp(t)\big), 
\end{equation}
holding for all $(\stb,\gdb)\in\proj_\mr{x,u}\B$ and for all $t\in\nninteg$. Substituting the differential dynamics \cref{7_an_eq:sys_diff}, the considered supply function \cref{7_an_eq:supply_diff}, and storage function \cref{7_an_eq:diffStorage} in \cref{7_an_eq:proofdiffdissip1} results in 
\begin{multline}\label{7_an_eq:proofdiffdissip2}
	\!\!\!\!\!(\star)\!^\top\!\storquad(\stb(t+1))\big(\difA(\stb(t),\gdb(t))\dst(t)\!+\!\difB(\stb(t),\gdb(t))\dgd(t)\big) - \\ \dst(t)^\top \storquad(\stb(t))\dst(t) \leq \dgd(t)^\top \supQ \dgd(t) +\\ 2 \dgd(t)^\top \supS \big(\difC(\stb(t),\gdb(t))\dst(t) + \difD(\stb(t),\gdb(t))\dgd(t)\big)
	+\\ (\star)^\top  \supR \big(\difC(\stb(t),\gdb(t))\dst(t) + \difD(\stb(t),\gdb(t))\dgd(t)\big), 
\end{multline}
holding for all $(\stb,\gdb)\in\proj_\mr{x,u}\B$ and for all $t\in\nninteg$. If it holds for all $(\stb,\gdb)\in{\stSet}\times{\gdSet}$, $\stdot\in\stdotSet$, $\dst\in\reals^\stSize$, and $\dgd\in\reals^\gdSize$ that 
\begin{multline}\label{7_an_eq:proofdiffdissip3}
	(\star)^\top \storquad(\stb+\stbdot)\big(\difA(\stb,\gdb)\dst+\difB(\stb,\gdb)\dgd\big) - \\ \dst^\top \storquad(\stb)\dst \leq \dgd^\top \supQ \dgd + 2 \dgd^\top \supS \big(\difC(\stb,\gdb)\dst + \\ \difD(\stb,\gdb)\dgd\big)
	+ (\star)^\top  \supR \big(\difC(\stb,\gdb)\dst + \difD(\stb,\gdb)\dgd\big), 
\end{multline}
then, \cref{7_an_eq:proofdiffdissip2} holds. Finally, \cref{7_an_eq:diffDissipFull} is equivalent to \cref{7_an_eq:proofdiffdissip3} by pre- and post multiplication of \cref{7_an_eq:diffDissipFull} with $\col(\dst, \dgd)^\top$ and $\col(\dst, \dgd)$, respectively.

\proofsection{4_lem:difstab}\label{4_pf:difstab}
The proof follows in a similar manner as \cref{4_lem:velostab}. Namely, let the system given by \cref{8_eq:nonlinsys} be differentially dissipative w.r.t. a supply function $\supfunDif$. For this supply function, \cref{8_eq:supplystabilitydiff} holds for all $\dgp\in\reals^\gpSize$. Therefore, it holds that
\begin{multline}\label{4_eq:storfunshiftproofdif}
\storfunDif(\stb(t_1+1),\dst(t_1+1))-\storfunDif(\stb(t_0),\dst(t_0))\leq \\\sum_{t=t_0}^{t_1} \supfunDif(\dgd(t),\dgp(t)) \leq 0
\end{multline}
for all $t_0,t_1\in\nninteg$ with $t_0\leq t_1$ and $(\dtst,\dtgd,\dtgp)\in\Bv$. This gives that
 \begin{equation}\label{4_eq:storfunshiftproof2dif}
	\storfunDif(\stb(t+1), \dst(t+1)) -\storfunDif(\stb(t+1), \dst(t+1)) \leq 0,
 \end{equation}
	for all $t\in\nninteg$ and $\stb\in\proj_{\st}\Bw(\gd)$ under all measurable and bounded $\gd\in(\reals^\gdSize)^\nninteg$. Moreover, the storage function $\storfunDif$ satisfies the conditions for the function $\lyapfunDif$ in \cref{7_an_lem:diffstabcondmi}. Hence, by \cref{7_an_lem:diffstabcondmi}, \cref{4_eq:storfunshiftproof2dif} implies that the  system is IS.

In case of IAS, the supply function satisfies \cref{8_eq:supplystabilitydiff}, but with strict inequality for all $\dgp\in\reals^\gpSize$ when $\st(t)\neq\sto(t)$. Therefore, \cref{4_eq:storfunshiftproof2} holds, but with strict inequality when $\st(t)\neq\sto(t)$, which by \cref{7_an_lem:diffstabcondmi} implies USAS.

\proofsection{7_an_thm:ind-incr-dissip}\label{7_an_pf:ind-incr-dissip}
To prove our result, we will make use of the results in the proof of \cite[Theorem 6]{Verhoek2020}. As the system is differentially dissipative, it implies that by writing out the $\var$-dependence and integrating over $\var$,
\begin{multline}\label{7_an_eq:jib1}
\!\!\!\!\int_0^1\!\Big[\storfunDif\big(\stb(t_1+1,\var),\dst(t_1+1,\var)\big) - \storfunDif\big(\stb(t_0,\var),\dst(t_0,\var)\big)  -\\ \sum_{t=t_0}^{t_1}\supfunDif\big(\dgd(t,\var),\dgp(t,\var)\big)\Big]\,d \var \leq 0,
\end{multline}
holds for all $(\stb,\gdb)\in\proj_\mr{x,u}\B$, $\var\in[0,\,1]$, and for all $t_0,t_1\in\nninteg$ with $t_0 \le t_1$. Let us first consider the storage function part of this inequality. Let us define a minimum energy path between $\st$ and $\sto$:
\begin{equation}\label{5_eq:geodesic}
\stbmin_{(\st,\sto)}(\var):= \arginf_{\hat{\st}\in\pathSet(\st,\sto)}\int_0^1\storfunDif\left(\hat{\st}(\var),{\Partial{\hat{\st}(\var)}{\var}}\right)\,d \var.
\end{equation}
As $V_\delta(\stb,\dst) = \dst^\top \storquad(\stb)\dst$, the path $\stbmin_{(\st,\sto)}$ corresponds to the geodesic connecting $\st$ and $\stb$ under the Riemannian metric $\storquad(\stb)$, see also \cite{Manchester2018,ReyesBaez2019}. By the Hopf-Rinow theorem, this implies, for any $\st,\stb\in\reals^\stSize$, that $\stbmin_{(\st,\sto)}$ is a unique, smooth function \cite{Manchester2018,Manchester2017}. Based on this minimum energy path, we define the incremental storage function as: 
\begin{equation}\label{5_eq:jibstorage}
	\storfunIncr(\st,\sto):= \int_0^1\storfunDif\left(\stbmin_{(\st,\sto)}(\var),{\Partial{\stbmin_{(\st,\sto)}(\var)}{\var}}\right)\,d \var.
\end{equation}
Note that $\storfunDif(\stb,\cdot)\in\posClass{0},\,\forall\,\stb\in\reals^\stSize$. Therefore, $\storfunIncr(\st,\st)=0$ for all $\st\in\reals^\stSize$ as $\stbmin_{(\st,\st)}(\lambda) = \st$, hence, ${\Partial{\stbmin_{(\st,\sto)}(\var)}{\var}}=0$ and by definition $\storfunDif(\cdot,0)= 0$. Moreover, for all $\st,\sto\in\reals^\stSize$ for which $\st\neq \sto$, we have that $\storfunIncr(\st,\sto)>0$, as in that case there exists a set of $\var\in[0,1]$ for which ${\Partial{\stbmin_{(\st,\sto)}(\var)}{\var}}\in\reals\backslash\{0\}$ (as it can only be zero for all $\var$ if $\st=\sto$) and by definition $\storfunDif(\stb,\dst)>0,\,\forall\dst\in\reals^\stSize\backslash\{0\}$. Consequently, we have that $\storfunIncr\in\posClassI$.

Based on the definition of the incremental storage function, it follows that
\begin{equation}\label{5_eq:jib-1}
	\storfunIncr(\st(t_1\!+\!1),\sto(t_1\!+\!1))\!\leq\! \int_0^1 \!\storfunDif\big(\stb(t_1\!+\!1,\var),\dst(t_1\!+\!1,\var)\big)\,d \var,
\end{equation}
for any $(\var\mapsto\stb(t_1\!+\!1,\var))\in \pathSet(\st(t_1\!+\!1),\sto(t_1\!+\!1))$ with $\st(t_1\!+\!1),\,\sto(t_1\!+\!1)\in\reals^\stSize$, $t_1\in\nninteg$, and  $(t\mapsto\stb(t,\var))\in\proj_\mr{\st} \B$ for any $\var \in[0,\,1]$.
Moreover, we parameterize the initial condition as $\stb(t_0,\var) = \parTraj \stIc(\var) = \stbmin_{(\stIc,\otherTraj \stIc)}(\var)$, from which it follows that
\begin{equation}\label{5_eq:jib0}
-\storfunIncr(\st(t_0),\sto(t_0)) = -\int_0^1  \storfunDif\big(\stb(t_0,\var),\dst(t_0,\var)\big)\,d \var. 
\end{equation}
Combining \cref{5_eq:jib-1} and \cref{5_eq:jib0} gives that
\begin{multline}\label{7_an_eq:incrstor_difstor_bound}
	\storfunIncr\big(\st(t_1+1),\sto(t_1+1)\big)-\storfunIncr\big(\st(t_0),\sto(t_0)\big)\leq\\
	\int_0^1\! \storfunDif\big(\stb(t_1\!+\!1,\var),\dst(t_1\!+\!1,\var)\big) \!-\! \storfunDif\big(\stb(t_0,\var),\dst(t_0,\var)\big) \,d \var.
\end{multline}

Subsequently, we consider the supply function part of \cref{7_an_eq:jib1}. This follows in the same manner as in \cite{Koelewijn2021a, Verhoek2020}, which we will briefly repeat. By changing summation and integration operations, the supply function part  of \cref{7_an_eq:jib1} is given by 
\begin{equation}\label{7_an_eq:jib3}
	\sum_{t=t_0}^{t_1}\int_0^1(\star)^\top\qsrMat\begin{bmatrix}\dgd(t,\var) \\ \dgp(t,\var)\end{bmatrix}\,d \var.
\end{equation}
Parameterizing our input as $\gdb(t,\var) =\gdo(t)+\var(\gd(t)-\gdo(t))$, it follows that $\dgd(t) = \Partial{\gdb(t,\var)}{\var}=\gd(t)-\gdo(t)$. Therefore, we have that $\int_0^1(\star)^\top \supQ\,\dgd(t,\var)d\var = (\star)^\top \supQ (\gd(t)-\gdo(t))$ and 
\begin{equation}
\begin{aligned}
\int_0^1 \!2\,\dgd(t,\var)\!^\top \!S\,\dgp(t,\var)d\var &= 2(\gd(t)\!-\!\gdo(t))\!^\top \! \supS{\int_0^1  \Partial{\gpb(t,\var)}{\var}}\,d \var, \\
&= 2(\gd(t)\!-\!\gdo(t))\!^\top \! \supS\, (\gpb(t,1)\! -\!\gpb(t,0)), \\
&=2(\gd(t)\!-\!\gdo(t))\!^\top \! \supS\, (\gp(t) -\gpo(t)).
\end{aligned}
\end{equation}
As we consider $\supR\preceq0$, i.e., $-\supR\succeq 0$, we have by \cite[Lemma 16]{Koelewijn2021a}
\begin{multline}\label{5_eq:jibqsrend}
\int_0^1  (\star)^\top \supR \,\dgp(t,\var)d\var = 
\int_0^1  (\star)^\top \supR\, \Partial{\gpb(t,\var)}{\var} \,d \var \le\\
(\star)^\top \supR \left(\int_0^1  \Partial{\gpb(t,\var)}{\var}\,d \var\right) = (\star)^\top \supR\,(\gp(t)-\gpo(t)).
\end{multline}

Combining this, results in the following inequality to hold
\begin{multline}\label{7_an_eq:jib5}
	 \sum_{t=t_0}^{t_1}\int_0^1(\star)^\top\qsrMat\begin{bmatrix}\dgd(t,\var) \\ \dgp(t,\var)\end{bmatrix}\,d \var \leq\\ \sum_{t=t_0}^{t_1}(\star)^\top\qsrMat\begin{bmatrix}\gd(t)-\gdo(t)\\\gp(t)-\gpo(t)\end{bmatrix}.
\end{multline}
Combining \cref{7_an_eq:incrstor_difstor_bound,7_an_eq:jib5} with \cref{7_an_eq:jib1} results in 
\begin{multline}
{\storfunIncr}\big(\st(t_1+1), \sto(t_1+1)\big) - {\storfunIncr}\big(\st(t_0), \sto(t_0)\big) \le\\ \sum_{t=t_0}^{t_1}\supfunIncr \big(\gd(t), \gdo(t),\gp(t), \gpo(t)\big),
\end{multline}
for all $t_0,t_1\in\nninteg$ with $t_0 \leq t_1$ and any two trajectories $(\st,\gd,\gp),(\sto,\gdo,\gpo)\in\B$ where $\storfunIncr$ is given by \cref{5_eq:jibstorage}, which is the condition for incremental dissipativity in \cref{7_an_def:incrdis}.

\proofsection{7_an_thm:diffdissipLPV}\label{7_an_pf:diffdissipLPV}
We have that \cref{7_an_eq:lpvdiff} is a DPV embedding of \cref{8_eq:nonlinsys} on the region $\stSetLPV\times\gdSetLPV\,\subseteq\reals^\stSize\!\times\reals^\gdSize\!$. Therefore, $\lpvA(\schMap(\stb,\gdb))=\difA(\stb,\gdb)$, $\dots$, $\lpvD(\schMap(\stb,\gdb))=\difD(\stb,\gdb)$ for all $(\stb,\gdb)\in\stSetLPV\times\gdSetLPV\,\subseteq\reals^\stSize\!\times\reals^\gdSize\!$. Moreover, we have \cref{7_an_eq:diffDissipFullLPV} holds for all $\sch\in\schSet$ and $\schdot\in\schdotSet$. Hence, it straightforwardly follows that \cref{7_an_eq:diffDissipFull} holds for all $(\stb,\gdb)\in{\stSet}\times{\gdSet}$ and $\stbdot\in\stdotSet$.

\bibliographystyle{elsarticle-num}
\bibliography{bibtex_db}

\end{document}